\documentclass[12pt,reqno]{amsart}
\usepackage{a4wide}

\usepackage{amssymb,amsmath,amsthm,enumerate,amsfonts}
\usepackage{newlfont}
\usepackage{dsfont}
\usepackage{amsbsy}
\usepackage[dvips]{graphicx}
\usepackage{verbatim} 
\usepackage{float}
\usepackage{color}

\def\hcboxcm#1#2{\hbox to #1{\hfill #2 \hfill}}

\newcommand{\n}{\noindent}

\def\null{\hbox{}}

\def\tn1{\widetilde n_1}
\def\tn2{\widetilde n_2}
\def\tn{\widetilde n }

\let\ds\displaystyle

\def\be{\begin{equation}}
\def\ee{\end{equation}}
\def\bea{\begin{eqnarray}}
\def\eea{\end{eqnarray}}

\def\bean{\begin{eqnarray*}}
\def\eean{\end{eqnarray*}}

\def\={\, = \, }

\def\NN{{\mathbb N} }
 \def\OO{\rm \hbox{O\kern-.34em\raise.47ex
         \hbox{$\scriptscriptstyle |$}\kern-.46em\raise.47ex
         \hbox{$\scriptscriptstyle |$}\kern+0.5 em }}
\def\RR{{\mathbb R} }

\def\ZZ{{\mathbb Z} }

\def\NN{\mathbb{N}}
\def\OO{\mathbb{O}}

\def\RR{\mathbb{R}}

\def\ZZ{\mathbb{Z}}

\def\Box{\leavevmode\vbox{\hrule
     \hbox{\vrule\kern4pt\vbox{\kern4pt}%
           \vrule}\hrule}}
\def\blackbox{\leavevmode\vrule height 5pt width 4pt depth 0pt\relax}
\catcode`@=11

\def\eqalign#1{\null\,\vcenter{\openup1\jot \m@th
   \ialign{\strut \hfil$\displaystyle{##}$ & $\displaystyle{{}##}$\hfil
      \crcr#1\crcr}}\,}
\def\eqalignrll#1{\null\,\vcenter{\openup1\jot \m@th
   \ialign{\strut \hfil$\displaystyle{##}$ & $\displaystyle{{}##}$\hfil
    & $\displaystyle{{}##}$\hfil
      \crcr#1\crcr}}\,}
\def\eqalignrcl#1{\null\,\vcenter{\openup1\jot \m@th
   \ialign{\strut \hfil$\displaystyle{##}$ &\hfil $\displaystyle{{}##}$\hfil
    & $\displaystyle{{}##}$\hfil
      \crcr#1\crcr}}\,}
\def\eqalignno#1{\displ@y \tabskip\@centering
  \halign to\displaywidth{\hfil$\@lign\displaystyle{##}$\tabskip\z@skip
    &$\@lign\displaystyle{{}##}$\hfil\tabskip\@centering
    &\llap{$\@lign##$}\tabskip\z@skip\crcr
    #1\crcr}}
\newcounter{appendix}
\newcounter{sectionz}
\def\appendix{\advance\c@appendix by 1
\def\thesectionz {\Alph{appendix}}
\def\thesection{\Alph{appendix}} 
   \ifnum\c@appendix=1 \setcounter{section}{-1} \fi
   \@startsection {section}{1}{\z@}{-3.5ex plus -1ex minus 
  -.2ex}{2.3ex plus .2ex}{\large\bf}}

\catcode`@=12
\newtheorem{lemme}{Lemma}[section]  

\newtheorem{theorem}[lemme]{Theorem}

\newtheorem{corollary}[lemme]{Corollary}

\newtheorem{definition}[lemme]{Definition}

\newtheorem{proposition}[lemme]{Proposition}

\newtheorem{remark}[lemme]{Remark} 

\def\deblem{\begin{lemme}\it }
\def\finlem{\end{lemme}}
\def\debthm{\begin{theorem}\it }
\def\finthm{\end{theorem}}
\def\debprop{\begin{proposition} \it}
\def\finprop{\end{proposition}}
\def\debcor{\begin{corollary}\it }
\def\fincor{\end{corollary}}
\def\debdef{\begin{definition}\it}
\def\findef{\end{definition}}
\def\debrem{\begin{remark}\em}
\def\finrem{\null\hfill\blackbox\end{remark}}

\title[The quantum beating]{The quantum beating and its numerical simulation}
\author[R. Carlone]{Raffaele Carlone}
\address{Universit\`{a} ``Federico II'' di Napoli, Dipartimento di Matematica e Applicazioni ``R. Caccioppoli'', MSA, via Cinthia, I-80126, Napoli, Italy.}
\email{raffaele.carlone@unina.it}

\author[R. Figari]{Rodolfo Figari}
\address{Universit\`{a} ``Federico II'' di Napoli, Dipartimento di Fisica e INFN Sezione di Napoli, MSA, I-80126, Napoli, via Cinthia, Italy.}
\email{rodolfo.figari@na.infn.it}

\author[C. Negulescu]{ Claudia Negulescu}
\address{Universit\'e de Toulouse \& CNRS, UPS, Institut de Math\'ematiques de Toulouse UMR 5219, F-31062 Toulouse, France}
\email{claudia.negulescu@math.univ-toulouse.fr}

\begin{document}
\maketitle

\begin{abstract}
We examine the suppression of quantum beating in a one dimensional non-linear double well  potential, made up of two focusing nonlinear point interactions.  The investigation of the Schr\"odinger dynamics is reduced to the study of a system of coupled nonlinear Volterra integral equations. For various values of the geometric and dynamical parameters of the model we give analytical and numerical results on the way states, which are initially confined in one well, evolve. We show that already for a nonlinearity exponent well below the critical value there is  complete suppression of the typical beating behavior characterizing the linear quantum case. 	
\end{abstract}

\bigskip




\section[Introduction]{Introduction}
 Quantum beating may nowadays refer to many, often quite different, phenomena studied in various domains of quantum physics, ranging from quantum electrodynamics to particle physics, from solid state physics to molecular structure and dynamics. \\
A paradigmatic example in the latter field is the inversion in the ammonia molecule observed experimentally in 1935. The ammonia molecule is pyramidally shaped. Three hydrogen atoms form the base and the nitrogen atom is located in one of the two distinguishable states (enantiomers) on one side or the other with respect to the base (chirality). Experimentally it was tested that microwave radiation could  induce a periodic transition from one state to the other (quantum beating). It was also observed that in several circumstances the pyramidal inversion was suppressed.  In particular in an ammonia gas the transition frequency was recognized to decrease for increasing pressure and the beating process proved to be finally suppressed for pressures above a critical value.  \\
A theoretical explanation of the quantum beating phenomenon was obtained modeling the nitrogen atom (better, the two ``non-bonding`` electrons of nitrogen) as a quantum particle in a double well potential.\\ 
The double well potential is of ubiquitous use in theoretical physics. In our present context, its importance consists in the fact that, for particular values of the parameters, the ground state and the first excited state have very close energies, forming an almost single, degenerate, energy level. A superposition of these two states is shown to evolve concentrating  periodically inside one well or the other, with a frequency proportional to the energy difference (see section \ref{SEC21} below).\\

According to the mathematical quantum theory of molecular structure developed in the second half of the last century (see \cite{Davies:1979ep,Davies:1995gb} and references therein; see also \cite{JonaLasinio:2001gr,Herbauts:2007ix} for studies of the pressure dependent transition mechanism) the effect of the ammonia molecule quantum environment can be modeled as a non-linear perturbation term added to the double well potential. A detailed quantitative analysis of the physical mechanism giving rise to the non-linear reaction of the environment, in the case of pyramidal molecules, can be found in \cite{Claverie:1986ee}.\\
Following this suggestion, Grecchi, Martinez and Sacchetti (\cite{Grecchi:1996de,Grecchi:2002bh,Sacchetti:2004fp,Sacchetti:2006io}) investigated the semi-classical limit of solutions to the non-linear Schr\"odinger initial value problem obtained perturbing the double well potential with a non-linear term breaking the rotational symmetry. They were able to prove that there is a critical value of the  coupling constant, measuring the strength of the symmetry breaking non-linear term, above which the beating period goes to infinity meaning that the beating phenomenon is suppressed. \\
In this paper we present a different model of a similar physical situation. We consider a hamiltonian with two concentrated non-linear attractive point potentials and we investigate the corresponding Cauchy problem. The study of the evolution problem is  reduced to the analysis 
of a system of two Volterra integral equations whose solutions we examine via numerical computation. It is worth stressing that the non-linear model we consider is governed by symmetric dynamical equations. Asymmetry will appear only as a consequence of the non-linearity and of specific initial conditions. We will come back to this point in the conclusions.\\
In Sections \ref{SEC21} and \ref{SEC22} we recall the properties of the corresponding symmetric and asymmetric linear case in order to clarify the origin of the beating phenomenon and its destruction. Due to the great degree of solvability of point interaction hamiltonians, the characterization of the beating states as functions of the dynamical and geometrical parameters of the model, will be carried through at a high level of detail.\\
In Section \ref{SEC3} we investigate, via numerical studies, the evolution problem in the linear symmetric, linear asymmetric and non-linear  cases. We show that the asymmetry resulting from the non-linearity causes beating suppression and a  rapid localization in one of the wells as soon as the non-linearity becomes relevant.\\
In the conclusions we compare our results with what is known in literature and we list some open problems and some possible extensions of our results.

\section{The mathematical model - Concentrated nonlinearities} \label{SEC2} 
In order to introduce the double well potential model  we shall investigate in the present paper, we first briefly recall the definition of point interaction hamiltonians in $L^{2}(\mathbb{R})$ (see (\cite{SAlbeverio:2011vta}) for further details).\\
For two point scatterers placed in $Y=\{y_{1},y_{2}\}$ of strength  $\underline{\gamma} = \{\gamma_{1}, \gamma_{2}\} ,\,\,\, \gamma_{i} \in \mathbb{R}$, the formal hamiltonian reads 
\begin{equation} \label{formal}
H_{\underline{\gamma}, Y}\, 
\psi :=`` - \frac{d^{2}}{dx^{2}} \, \psi + \gamma_1 \delta_{y_1} \psi + \gamma_2 \delta_{y_2} \psi\,``,
\end{equation}
where the (reduced) Planck constant $\hbar$ has been taken equal to one  and the particle mass $m$ equal to $1/2$.\\
A rigorous definition of $H_{\underline{\gamma}, Y}$ in dimension $d=1$ has been given in the early days of quantum mechanics, when such kind of hamiltonians were extensively used to investigate the dynamics of a quantum particle in various kinds of short range scatterer arrays. A complete characterization of point interaction hamiltonians in $3 \geq d >1$ was only made available in the second half of last century (see \cite{SAlbeverio:2011vta} for details and for an exhaustive bibliography).

\noindent
Restricted to the case of our interest, definition and main results in $d=1$ are shortly detailed below. Assume that the two points are placed symmetrically with respect to the origin and that $|y_{i}|=a$.  Then 

\begin{equation} \label{dom}
D(H_{\underline{\gamma}, Y}) := \Big\{ \psi \in L^{2}(\mathbb{R}) \;| \; \psi =
\phi^{\lambda} \,-\,   \sum_{i,j=1}^{2} \left( \Gamma^{\lambda}_{\underline{\gamma}} \right)^{-1}_{ij} \phi^{\lambda}(y_{j}) G^{\lambda}
(\cdot - y_{i}), \; \phi^{\lambda} \in H^{2}(\mathbb{R}) \Big\} \,,
\end{equation}
\be\label{oper}
\left(H_{\underline{\gamma}, Y} + \lambda\right) \psi = \left(-  \frac{d^{2}}{dx^{2}} + \lambda\right) \phi^{\lambda},
\ee
are domain and action of a selfadjoint operator in $L^{2}(\mathbb{R})$ which acts as  the free laplacian on functions supported outside the two points $y_{i}=\pm a$.
\noindent
In (\ref{dom}) $G^{\lambda}(\cdot)$ is the Green function for the free laplacian, given by
\begin{equation}\label{GG}
G^{\lambda}(x):=\frac{e^{-\sqrt{\lambda}|x|}}{2\sqrt{\lambda}},
\end{equation}
and the matrix $\Gamma^{\lambda}_{\underline{\gamma}}$ is defined as
\be \label{bcond}
 \left( \Gamma^{\lambda}_{\underline{\gamma}} \right)_{ij}\ :=  \frac{1}{\gamma_{i}}  \, \delta_{ij}  + G^{\lambda} (y_{i} - y_{j})\,,
\ee
where the positive real number $\lambda$ is chosen large enough to make the matrix $\Gamma^{\lambda}_{\underline{\gamma}}$ invertible.

\noindent
It is immediate to check that the derivative of $G^{\lambda}(x)$ has a jump in the origin, equal to $-1$. This in turn implies that every function $\psi\,$ in the domain satisfies the boundary conditions  
\be \label{bc}
\frac{d \psi}{dx}\left(y_{j}^{+}\right) -\frac{d \psi}{dx}\left(y_{j}^{-}\right)  = \gamma_{j} \,\psi (y_{j})\,, \quad j=1,2\,.
\ee
The dynamics generated by $H_{\underline{\gamma}, Y}$ is then characterized as the free dynamics outside the two scatterers, satisfying at any time the boundary conditions (\ref{bc}).
 
\noindent
Our aim is to investigate the behaviour of the solutions to the non-autonomous evolution problem
\begin{equation}\label{SCH}
\left\{
\begin{array}{l}
\ds {\imath}\, \frac{\partial \psi}{\partial t}  = H_{\underline{\gamma} (t), Y}\, \psi\,, \quad \forall (t,x) \in \RR^+ \times \RR\,,\\[3mm]
\ds \psi(0,x)=\psi_0(x) \in D(H_{\underline{\gamma}(0), Y})\quad \forall x \in \RR\,, \\[3mm]
\ds \gamma_j(t):= \gamma |\psi(t,y_j)|^{2 \sigma}, \,\,\,\gamma < 0, \,\,\, \sigma \geq 0.
\end{array}
\right.
\end{equation}
where the time dependence of $\underline{\gamma}$ is non-linearly determined by the values in $ \pm a$ of the solution itself.


\noindent
An alternative way to examine the Cauchy problem (\ref{SCH}) is to write down Duhamel's formula corresponding to the formal Hamiltonian  (\ref{formal}) with the coupling constants $\underline{\gamma}$ given in (\ref{SCH}), then prove that the boundary conditions are satisfied at each time.

\noindent
In detail, let $U(\tau, y)$ be the integral kernel of the unitary group $\displaystyle e^{\imath t \Delta}$
$$
U(\tau,y):={e^{{\imath}{|y|^2 \over 4\tau}} \over \sqrt{4 {\imath}\, \pi\, \tau}}\,, \quad (\mathcal{U}(t)\xi)(x)= \int_{-\infty}^\infty U(t;x-y)\, \xi(y)\, dy\,\,\,\,\,\,\,\,\,\,\,\, \,\,\,\,\,\,\,\, \forall \xi \in L^{2}(\RR).
$$
Then from the ansatz
\begin{equation}\label{duhamel_0}
\psi(t,x)=(\mathcal{U}(t)\psi_{0})(x)-{\imath}\,\gamma\sum_{j=1}^{2}\int_{0}^{t} U(t-s;x-y_{j})|\psi(s,y_{j})|^{2\sigma}\psi(s,y_{j})\, ds\,,
\end{equation}
one obtains for $i=1,2$
\begin{equation}\label{duhamel}
\psi(t,y_{i})=(\mathcal{U}(t)\psi_{0})(y_{i})-{\imath}\,\gamma\sum_{j=1}^{2}\int_{0}^{t} U(t-s;y_{i}-y_{j})|\psi(s,y_{j})|^{2\sigma}\psi(s,y_{j})\, ds.
\end{equation}
Explicitly 
 
\be \label{VOLT}
\left\{
\begin{array}{l}
\ds \psi(t,-a)+{\gamma \over 2} \sqrt{{\imath} \over  \pi} \, \int_0^t  { \psi(s,-a)\, |\psi(s,-a)|^{2 \sigma} \over \sqrt{t-s}}\, ds + {\gamma \over 2} \sqrt{{\imath} \over  \pi} \, \int_0^t  { \psi(s,a)\, |\psi(s,a)|^{2 \sigma} \over \sqrt{t-s}} \, e^{{\imath} { a^2 \over {(t-s)}}}\, ds\\[3mm]
\ds \hspace{10cm} = ({\mathcal U}(t)\, \psi_0)(-a)\,,\\[3mm]
\ds \psi(t,a)+{\gamma \over 2} \sqrt{{\imath} \over  \pi} \, \int_0^t  { \psi(s,a)\, |\psi(s,a)|^{2 \sigma} \over \sqrt{t-s}}\, ds + {\gamma \over 2} \sqrt{{\imath} \over  \pi} \, \int_0^t  { \psi(s,-a)\, |\psi(s,-a)|^{2 \sigma} \over \sqrt{t-s}} \, e^{{\imath} { a^2 \over {(t-s)}}}\, ds\\[3mm]
\ds \hspace{10cm} = ({\mathcal U}(t)\, \psi_0)(a)\,.
\end{array}
\right.
\ee
It is easy to check that a function of the form (\ref{duhamel_0}) satisfies the non-linear boundary conditions at all times (see 
\cite{Adami:2001bt} for details). Following a standard use  in higher dimensional cases,   we will often employ in this paper the notation $q_{1}(t)  \equiv \psi(t,-a) , \,\,q_{2}(t)  \equiv \psi(t,a)$ and refer to (\ref{VOLT}) as the ``charge equations''. \\ 

\vspace{.2cm}
\noindent
The Cauchy problem (\ref{SCH}) is then reduced to the computation of $({\mathcal U}(t)\, \psi_0)(\pm a)$ and the solutions of the system  \eqref{VOLT}, corresponding to two coupled nonlinear Volterra integral equations. The whole wave-function is then recovered via \eqref{duhamel_0}.\\

\noindent

In the following we will show that the linear case $\sigma =0$ is characterized by the presence of almost stationary states whose wave function evolves periodically between one well and the other (beating states).  Along the lines traced by many authors in the past (see \cite{Davies:1979ep},\cite{Davies:1995gb} and \cite{Grecchi:2002bh}), we will then show that the nonlinearity destroys the beating phenomenon. The reduction in complexity we obtain, using linear and non-linear point interactions, makes the investigation of the theoretical and computational aspects of the problem easier. In order to better understand how the beating effect occurs and the reasons why one expects suppression of beating by nonlinear perturbation, we develop in Sections \ref{SEC21} and 
\ref{SEC22}  the symmetric and antisymmetric linear cases in some detail.\\

\subsection{Linear point interactions - Symmetric double well} \label{SEC21} 

\noindent
Let us consider the symmetric linear case, corresponding to $\sigma=0$ and $\gamma_1=\gamma_2=\gamma$. We will show that the eigenstates relative to the lowest eigenvalues are explicitly computable for the hamiltonian we consider.

\noindent
In fact, applying $(H_{\underline{\gamma}, Y} + \lambda)^{-1}$ to both sides of  (\ref{oper}) and using (\ref{dom}) we obtain that  for all $\xi \in H^{2}(\RR)$ one has
\begin{equation}
(H_{\underline{\gamma}, Y} + \lambda)^{-1} \xi = \left(-  \frac{d^{2}}{dx^{2}} + \lambda\right)^{-1} \xi -  \sum_{i,j=1}^{2} \left( \Gamma^{\lambda}_{\underline{\gamma}} \right)^{-1}_{ij}  \left[\left(-  \frac{d^{2}}{dx^{2}} + \lambda\right)^{-1} \xi \right] (y_{j}) G^{\lambda}
(\cdot - y_{i}) \nonumber \\
\end{equation}

\n
which implies that the integral kernel of the resolvent is 
\be\label{respi}
(H_{\underline{\gamma}, Y} +\lambda)^{-1} (x,x')= G^{\lambda}(x - x') - \sum_{i,j =1}^2  \left( \Gamma^{\lambda}_{\underline{\gamma}} \right)^{-1}_{ij} G^{\lambda}(x - y_i) G^{\lambda} (x' - y_j).
\ee

\noindent
As it is clear from (\ref{respi}), the resolvent of $H_{\underline{\gamma}, Y}$  is a finite rank perturbation of the free laplacian resolvent operator. From the kernel representation of the resolvent, spectral  and scattering properties of the operator $H_{\underline{\gamma}, Y}$ are easily inquired in the case of  interactions of equal strength  $\gamma_{i}=\gamma$ (see \cite{SAlbeverio:2011vta}, Theorem 2.1.3).

\noindent
Only the second term appearing in the formula for the resolvent (\ref{respi}) can have polar singularities for those positive values of $\lambda$ for which the matrix $\Gamma^{\lambda}_{\underline{\gamma}} $ is not invertible . In particular,  $- \lambda$ will be a negative eigenvalue of $H_{\underline{\gamma}, Y}$ if and only if
 \be \label{determinant}
\textrm{det} \,   \Gamma^{\lambda}_{\underline{\gamma}}  = 0.
\ee
In the case of two point interactions of the same strength this condition reads 
\begin{equation}\label{det}
\textrm{det}\, \left( \begin{matrix}\frac{1}{\gamma}+\frac{1}{2\sqrt{\lambda} }& G^{\lambda}(2 a)\\G^{\lambda}(2 a)&\frac{1}{\gamma}+\frac{1}{2\sqrt{\lambda}}\end{matrix} \right)=0.
\end{equation}
For $\gamma<-\frac{1}{a}$ there are two solutions $\lambda_{f,e}>0$ to the previous equation. The indices $"f,e"$ stand for ``fundamental'' resp. first ``excited'' state. The corresponding eigenfunctions are

\begin{equation}\label{PF}
\phi_{f}(x)=N_{f}\left(G^{\lambda_f}(x+a)+G^{\lambda_f}(x-a)\right)
\end{equation}
\begin{equation}\label{PE}
\phi_{e}(x)=N_{e}\left(G^{\lambda_{e}}(x+a)-G^{\lambda_{e}}(x-a)\right)\,,
\end{equation}
where $N_{f}$ and $N_{e}$ are easily computable normalization factors.

In Fig. \ref{twoen} we plotted the two eigenstates $\phi_f(x)$ and $\phi_e(x)$, corresponding to the fundamental state (symmetric function) and the first excited state (anti-symmetric function). Notice that the two eigenstates are relative to energies getting closer and closer as the value of $|\gamma|$ increases (see the remark below). In the same limit the absolute values of the two eigenfunctions tend to coincide.\\
\begin{figure}[H]
    \centering
\includegraphics[width=8cm]{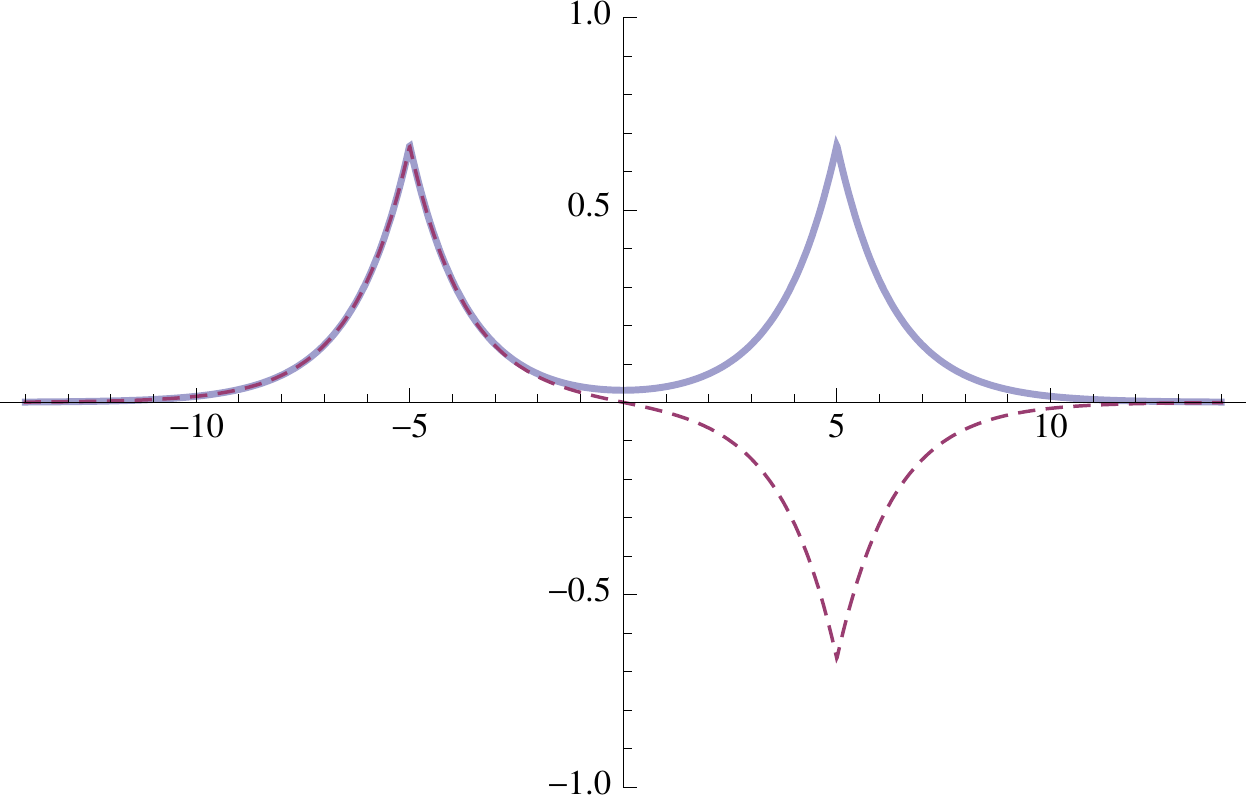}
 \caption{Plot of the functions $\phi_{f}(x)$ with a  thicked blue line and $\phi_{e}(x)$ with a dashed line.}
    \label{twoen}
\end{figure}

\noindent
The stationary solutions corresponding to these eigenstates are given by
$$
\psi_f(t,x)=e^{{\imath}\, \lambda_f\, t} \phi_f(x)\,, \quad \psi_e(t,x)=e^{{\imath}\, \lambda_e\, t} \phi_e(x)\,.
$$

\noindent
A superposition of the two stationary states,
$$
\psi_0 (x)=\alpha\, \phi_f(x)+ \beta\, \phi_e(x)\,, \qquad \alpha,\, \beta \in \RR\,,
$$
evolves as
$$
\psi (t,x)=\alpha\, e^{{\imath}\, \lambda_f\, t} \phi_f(x)+ \beta\,e^{{\imath}\, \lambda_e\, t} \phi_e(x)\,.
$$ 
 
\noindent
In particular the superposition 
\be \label{beatINI}
\psi^{L}_{beat,0}(x):=\frac{1}{\sqrt{2}}\left(\phi_{f}(x)+\phi_{e}(x)\right)\,,
\ee
is concentrated in the left well and will evolve in time as 

\begin{equation}\label{lbeat0}
\psi_{beat}^{L}(t,x) =\frac{1}{\sqrt{2}}\left(e^{{\imath}\lambda_f t}\phi_{f}(x)+e^{{\imath}\lambda_e t}\phi_{e}(x)\right)\,,
\end{equation}
with a probability density  given by
\begin{equation}\label{dens}
\mathcal{P}(t,x)=\frac{1}{2}\left[|\phi_{f}(x)|^{2}+|\phi_{e}(x)|^{2}+2\, \phi_{f}(x)\phi_{e}(x)\cos\left( (\lambda_f-\lambda_e)t\right) \right]\,. 
\end{equation}

\noindent 
It is clear that $\psi_{beat}^{L}$ is an oscillating function with period $\displaystyle T_B=\frac{2\pi}{|\lambda_f-\lambda_e|} $  concentrated successively on the left and on the right well, justifying the definition of (\ref{lbeat0}) as a beating state.

\noindent
The values assumed by the function $\psi_{beat}^{L}(t,x)$ in the centers of the two wells evolve as follows (see Figure \ref{q})
\begin{equation}\label{lbeat}
\begin{array}{lll}
q_{1}^{L}(t) &\equiv&\ds  \psi_{beat}^{L}(t,-a)=\frac{1}{\sqrt{2}}\left(e^{{\imath}\lambda_f t}\phi_{f}(-a)+e^{{\imath}\lambda_e t}\phi_{e}(-a)\right)\\[3mm]
q_{2}^{L}(t) &\equiv&\ds \psi_{beat}^{L}(t,a)=\frac{1}{\sqrt{2}}\left(e^{{\imath}\lambda_f t}\phi_{f}(a)+e^{{\imath}\lambda_e t}\phi_{e}(a)\right)\,.
\end{array}
\end{equation}

\begin{figure}[H]
    \centering
\includegraphics[width=10cm]{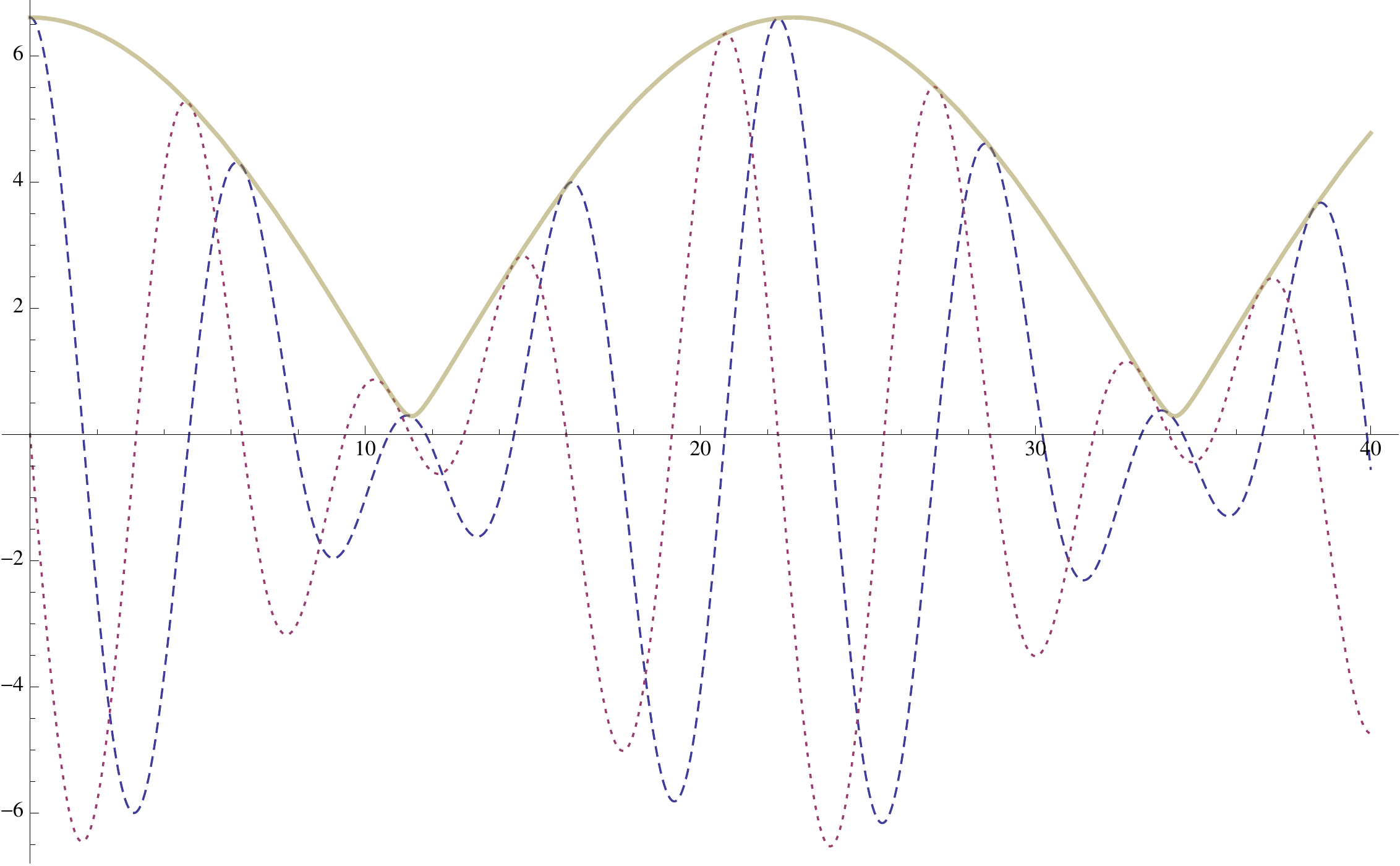}
 \caption{Plot of the time-evolution of the functions $\textrm{Re}\,q^{L}_{1}(t)$ as a dashed line, $\textrm{Im}\,q^{L}_{1}(t)$ as a dotted line and $|q^{L}_{1}|(t)$ as a thick line.}
    \label{q}
\end{figure}
\begin{remark} Many authors analyzed the energy difference between the fundamental and the first exited state of a hamiltonian with double well potential in the semi-classical limit, roughly referred to as $\hbar \rightarrow 0$ (see e.g. \cite{Claverie:1986ee}).  In the notes  \cite{Teta:2016vr} a detailed computation of the energy difference for a point interaction hamiltonian with two attractive zero range potentials is performed,  keeping all the standard dimensions of the physical constants. In terms of the dimensionless constant $\displaystyle \delta=\frac{2 m |\gamma| a}{\hbar^{2}}$ it is proved that in the limit $\delta \gg 1$
\begin{equation}
\triangle E \simeq \frac{2\,m\,\gamma^{2}}{\hbar^{2}}e^{-\delta}\,.
\end{equation}
The  exponential decay of the energy difference when $\hbar \rightarrow 0$ is then easily and rigorously obtained in the case of a zero range double well. Furthermore the result clarifies that  the semiclassical limit  is characterized by $\delta \rightarrow \infty$, which in our units reads $|\gamma| a \gg 1$.

\end{remark}
\subsection{Linear point interactions - Asymmetric double well} \label{SEC22} 
Let us now investigate the changes in the beating mechanism when the two zero range potentials have different strengths $\gamma_1 \neq \gamma_2$, $\gamma_1<\gamma_2 $.  In this asymmetric case the equation permitting to compute the eigenvalues of the Hamiltonian $H_{\underline{\gamma},\, Y}$ reads:

\begin{equation}\label{gammaij}
\det\Gamma_{(\gamma_{1},\gamma_{2})}^\lambda= \det\, \left(\begin{array}{cc}\frac{1}{\gamma_{1}}+\frac{1}{2\,\sqrt{\lambda}}& G^{\lambda}(2 a) \\ G^{\lambda}(2 a)  & \frac{1}{\gamma_{2}} + \frac{1}{2\,\sqrt{\lambda}}\end{array}\right)=0\,,
\end{equation}
leading to
\begin{equation}\label{eigen}
\left(\frac{1}{\gamma_{1}}+\frac{1}{2\,\sqrt{\lambda}}\right)\left(\frac{1}{\gamma_{2}}+\frac{1}{2\,\sqrt{\lambda}}\right)-\left(\frac{1}{2\,\sqrt{\lambda}}\right)^{2}e^{- 4 \sqrt\lambda \, a}=0\,.
\end{equation}

\noindent
All the relevant results we will need in the following are collected in the following lemma, concerning the resolution of this last equation.
\begin{lemme}\label{lemmino}
Let $\gamma_1 \neq \gamma_2$, $\gamma_1<\gamma_2 $ and let us define the ratio $\displaystyle \alpha := \frac{\gamma_{2}}{\gamma_{1}} $. Then one has:

\begin{description}
\item[a] There are two real solutions $\lambda_{0} > \lambda_{1}>0$ to equation (\ref{eigen}) if and only if $\gamma_{i} < 0$ for $i=1,2$ and
\begin{equation}\label{two}
\frac{1}{|\gamma_{1}|}+\frac{1}{|\gamma_{2}|}< 2a\,.
\end{equation}
\item[b]
For $\gamma_{i} < 0$, $i=1,2$, satisfying \eqref{two} and $\alpha <1$, one has
\[\Delta \lambda := \lambda_{0} - \lambda_{1} \geq \gamma_{1}^{2}  (1 -\alpha^{2})\,. \]
In particular $\Delta \lambda \rightarrow \infty$ as $|\gamma_{1}| \rightarrow \infty$\,.

\vspace{.1cm}
\item[c]  For $\gamma_{i} < 0$, $i=1,2$, satisfying \eqref{two} and $\alpha <1$, one has
\[\lim_ {|\gamma_{1}| \rightarrow \infty}    2 \sqrt{\lambda_{0}} / \gamma_{1} = -1\,, \quad \lim_ {|\gamma_{1}| \rightarrow \infty} 2 \sqrt{\lambda_{1}} /\gamma_{2} = -1\,.\]
\end{description}
\end{lemme}

\begin{proof}
Defining  $\xi:=2\sqrt{\lambda}$ equation (\ref{eigen}) can be rewritten as
\begin{equation}\label{f}
\frac{\xi^{2}}{\gamma_{1}\gamma_{2}}+\xi \left(\frac{1}{\gamma_{1}}+\frac{1}{\gamma_{2}}\right)+1=e^{-2 \xi a}\,.
\end{equation}
\noindent The number of positive solutions to (\ref{f}) depends on the values of the parameters $\gamma_{i}$ and on the distance $2a$. 


\noindent
\begin{description}
\item[a] For $\gamma_{1}<\gamma_{2}<0$  both the right and the left side of (\ref{f}) are convex functions of $\xi$, taking the common value $1$ when $\xi =0$. Furthermore, denoting by $P(\xi)$ the left hand side, we have that $P( |\gamma_{i}|)  = 0 < e^{-2 |\gamma_{i}| a}$ for $i =1,2 \,$. We deduce that there are two solutions to  (\ref{f}) if and only if the derivative of $P(\xi)$ for $\xi = 0$ is larger than $-2a$ (see Figure \ref{qq}).\\

\noindent
It is easy to check that if at least one of the $\gamma$ is positive there cannot be two positive solutions of (\ref{f}).


\begin{figure}[H]
 \centering
 \includegraphics[width=8cm]{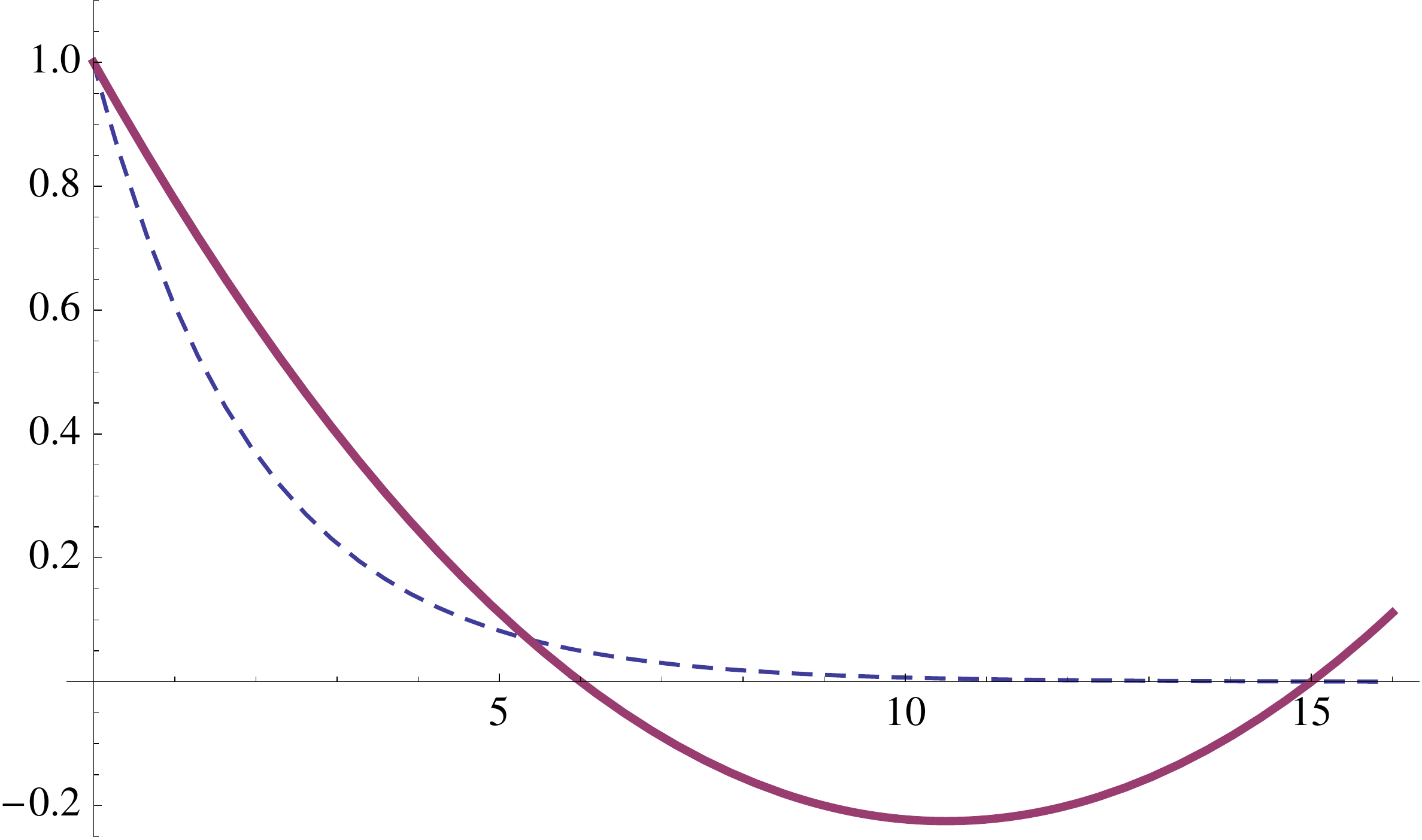}
 \caption{Plot of the functions $P(\xi)$ in red, $e^{-2 \xi a}$ dashed for $a=1/2$, $\gamma_{1}=-8$ and $\gamma_{2}=-4$.}
\label{qq}
\end{figure}


\item[b]
As it is clear from Figure \ref{qq}, the two solutions $\xi_{0} > \xi_{1}$ to (\ref{f}) are such that $\xi_{0} > \gamma_{1}$ and  $\xi_{1} < \gamma_{2}$. As a consequence

\be\label{delta}
 4 \Delta \lambda = \xi_{0}^{2} - \xi_{1}^{2} \, \geq \, \gamma_{1}^{2} - \gamma_{2}^{2}  = \gamma_{1}^{2} (1 - \alpha^{2}). 
\ee


\noindent
In the semi-classical regime $ |\gamma_{i}|\, a \gg 1$, $i=1,2$ (see remark at the end of previous subsection), (\ref{delta}) implies that the energy difference becomes larger and larger whenever $\alpha \neq 1$ (see Figure \ref{deltaE}). It is worth recalling that for $\alpha =1$ (symmetric case) the energy difference goes to zero in the same limit.\\

\item[c] 
Rewriting  (\ref{f}) in terms of $\eta := \xi/|\gamma_{1}|$ and $\alpha$ we obtain 
\begin{equation}\label{fe}
\frac{\eta^{2}}{\alpha} - \left(\frac{1 +\alpha}{\alpha}\right) \eta + 1=e^{-2 \eta |\gamma_{1}| a}\,.
\end{equation}
Both solutions $\eta_{0}$ and $\eta_{1}$ are strictly larger than zero. In turn, this implies that in the semi-classical limit $|\gamma_{i}|\, 2a \gg 1 \,\,\,\,i =1,2$, the exponential term  in (\ref{fe}) becomes negligible with respect to $1$ and $ \eta_{0} \rightarrow  1$, $\,\,\, \,\eta_{1}  \rightarrow  \alpha$.
\end{description}
\begin{figure}[H]
 \centering
\includegraphics[width=10cm]{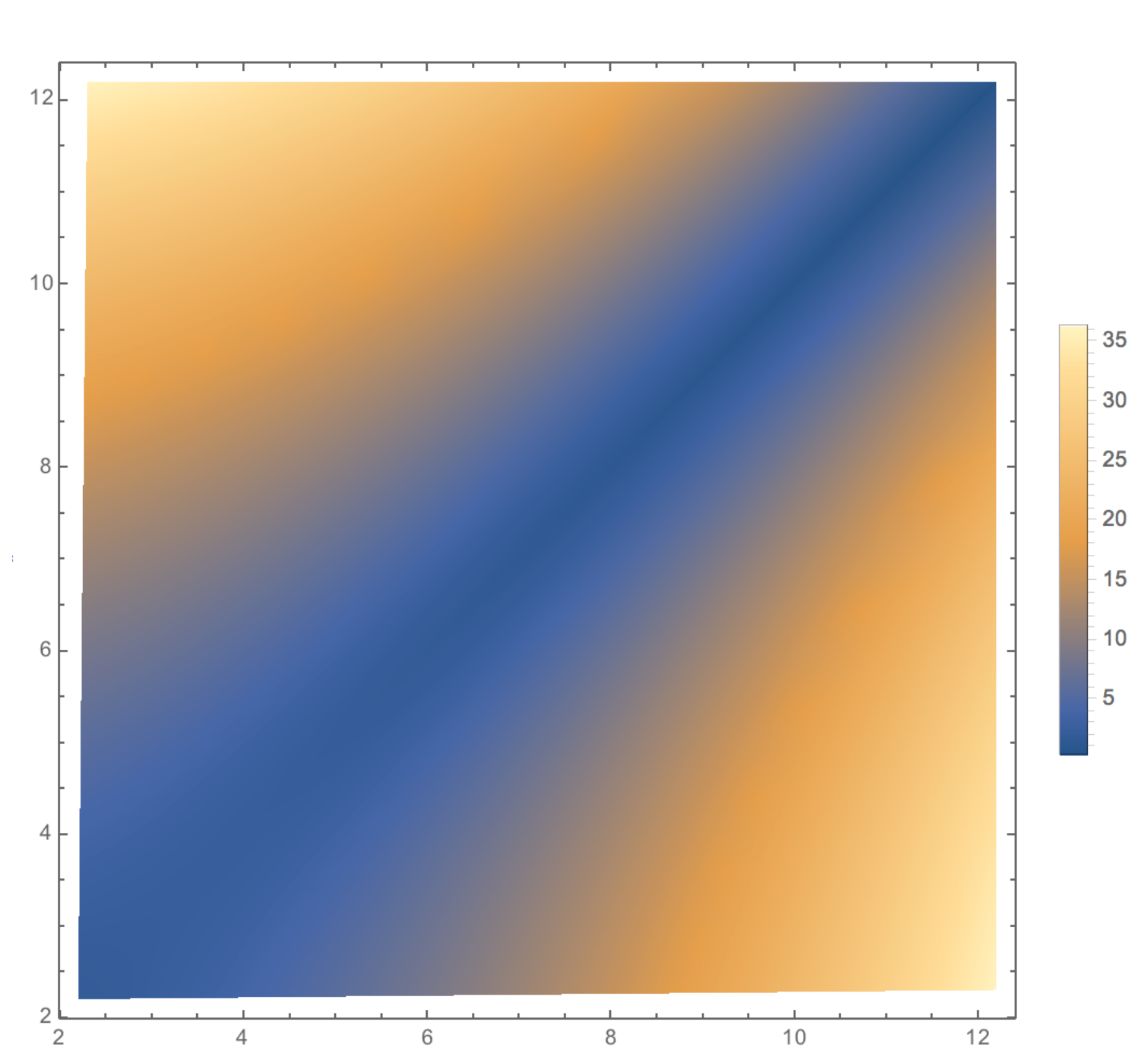}
 \caption{Contour plot of the functions $|\lambda_{0} - \lambda_{1}|$ solutions to the equation (\ref{eigen}) as a function of $|\gamma_{1}|$ and  $|\gamma_{2}|$ for $-12 \leq \gamma_{i} \leq -2\,\,\,\,\, i = 1,2$, $a=1$ .}
 \label{deltaE}
\end{figure}

\end{proof}

\noindent
The above results suggest that in the semi-classical limit with $\gamma_{1} \neq \gamma_{2}$ the fundamental eigenstate approaches the eigenstate of a single point interaction of strength $\gamma_{1}$ placed in $-a$ whereas the excited state approaches the eigenstate of a single point interaction of strength $\gamma_{2}$ placed in $a$.\\
In order to make this aspect clearer, we detail the steps needed to perform an exact computation of the eigenfunctions associated to the two eigenvalues.

\noindent
The normalized eigenfunction relative to the lowest eigenvalue  $\displaystyle E_0=- \lambda_{0}=-\frac{\xi_{0}^{2}}{4} < 0$ has the form (see again Theorem 2.1.3 in \cite{SAlbeverio:2011vta}) 

\begin{equation}\label{eigenfunction0}
\phi_0(x)=c_{0}G^{\lambda_{0}}(x-y_{1})+c_{1\,}G^{\lambda_{0}}(x-y_{2})\,,
\end{equation}
where $y_{1}=-a$, $y_{2}=a$ and the coefficient $c_{0},c_{1}$ are solutions of
\begin{equation}\left(\begin{array}{cc}\frac{1}{\gamma_{1}}+\frac{1}{2\,\sqrt{\lambda_0}}& \frac{1}{2\sqrt{\lambda_0}}e^{-2\,\sqrt{\lambda_0}\,a} \\ \frac{1}{2\sqrt{\lambda_0}}e^{-2\,\sqrt{\lambda_0}\,a}  & \frac{1}{\gamma_{2}} + \frac{1}{2\,\sqrt{\lambda_0}}\end{array}\right)\left(\begin{array}{c}c_{0}\\c_{1}\end{array}\right)=\left(\begin{array}{c}0\\0\end{array}\right)\,,
\end{equation}
which, together with (\ref{eigen}) gives
\begin{equation} \label{firsteig}
\left|\frac{c_{1}}{c_{0}}\right| =
  \sqrt{\frac{(2\sqrt{\lambda_{0}}/ \gamma_{1}) +1}{
(2\sqrt{\lambda_{0}}/ \gamma_{2}) +1}}\,. 
\end{equation}
The normalization condition finally gives for $c_{0}$
\begin{equation}\label{czero}
c_{0}=\frac{2 |\gamma _1| \lambda _0^{3/4}}
{\sqrt{\gamma _1\gamma _2\frac{\left(\gamma _1+2 \sqrt{\lambda _0}\right)}{\left(\gamma _2+2 \sqrt{\lambda _0}\right)} +\gamma _1 \left( \gamma _1+4 \sqrt{\lambda _0} +2 \sqrt{\lambda _0}\, a\left(\gamma _1 +2 \sqrt{\lambda _0}\right)\right)}}\,.
\end{equation}


\noindent
Under the assumptions we made on $\underline{\gamma}$, there will be a second eigenvalue $E_{1}>E_{0}$, $\,\,E_{1}=- \lambda_{1}<0$ whose corresponding normalized eigenfunction has the form
\begin{equation}\label{eigenfunction1}
\phi_1(x)=c_{2} G^{\lambda_{1}}(x-y_{1})+c_{3}G^{\lambda_{1}}(x-y_{2})\,,
\end{equation}
where 
\begin{equation} \label{seceig}
\left| \frac{c_{2}}{c_{3}} \right |=\, \sqrt{\frac{(2\sqrt{\lambda_{1}}/ \gamma_{2}) +1}{
(2\sqrt{\lambda_{1}}/ \gamma_{1}) +1
}}\,,
\end{equation}
with 
\begin{equation}\label{ctre}
c_{3}=\frac{2 |\gamma _2| \lambda _1^{3/4}}
{\sqrt{\gamma _1\gamma _2\frac{\left(\gamma _2+2 \sqrt{\lambda _1}\right)}{\left(\gamma _1+2 \sqrt{\lambda _1}\right)} -\gamma _2 \left( \gamma _2+4 \sqrt{\lambda _1} +2 \sqrt{\lambda _1}\, a\left(\gamma _2 +2 \sqrt{\lambda _1}\right)\right)}}\,.
\end{equation}

\vspace{.2cm}
\noindent
The thus obtained functions \eqref{eigenfunction0} resp. \eqref{eigenfunction1} are the eigenfunctions corresponding to the eigenvalues $\lambda_0,\lambda_1$. The initial condition we shall choose in the asymmetric case will be of the form
\be \label{AsyINIT}
\psi_{asy,0}(x):= \alpha\, \phi_0(x)+\beta\, \phi_1(x)\,, \quad \alpha,\beta \in \RR\,,
\ee
 where the exact time-evolution of this state is given by
 \be \label{EX_ASY}
 \psi_{asy}(t,x):= \alpha\, e^{{\mathbf i} \, \lambda_0\, t}\, \phi_0(x)+\beta\,e^{{\mathbf i} \, \lambda_1\, t}\, \phi_1(x)\,.
 \ee

Let us also remark here that from (\ref{firsteig}) and (\ref{seceig}) we deduce that both $\displaystyle \left|\frac{c_{1}}{c_{0}}\right|$ and 
$\displaystyle \left|\frac{c_{2}}{c_{3}} \right|$ become negligible in the semi-classical regime, if $\displaystyle\alpha  = \frac{\gamma_{2}}{\gamma_{1}} <1$.  Taking into account the normalization factors (\ref{czero}), (\ref{ctre}) and part {\bf (c)} of lemma (\ref{lemmino}) we finally obtain that the fundamental state tends to $\lambda_{0}^{3/4} G^{\lambda_{0}}(x+a)$ and the excited state to $\lambda_{1}^{3/4} G^{\lambda_{1}}(x-a)$(see Figure \ref{fond ecc}). 



\begin{figure}[H]
\begin{center}
\includegraphics[width=8cm]{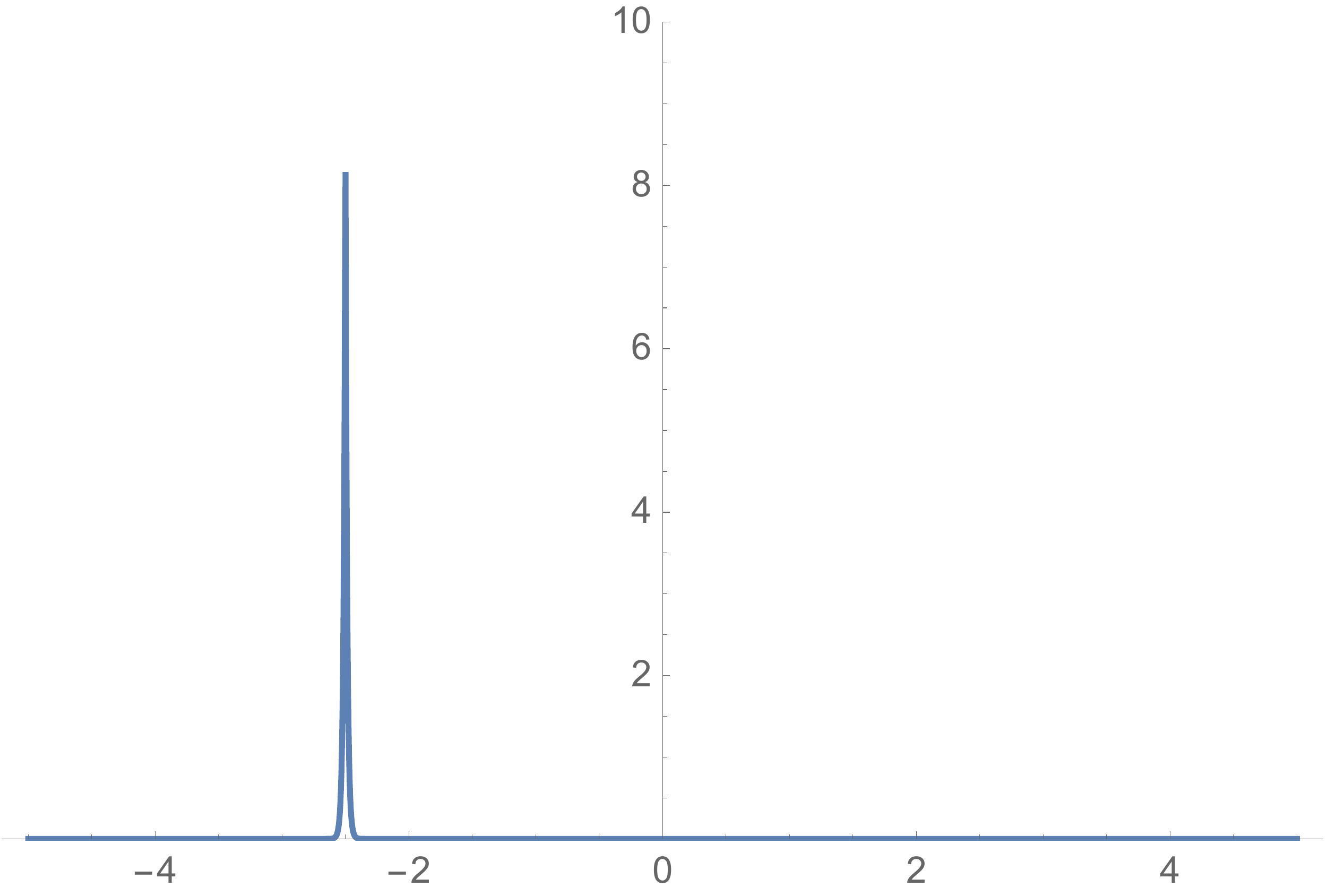}\hfill
\includegraphics[width=8cm]{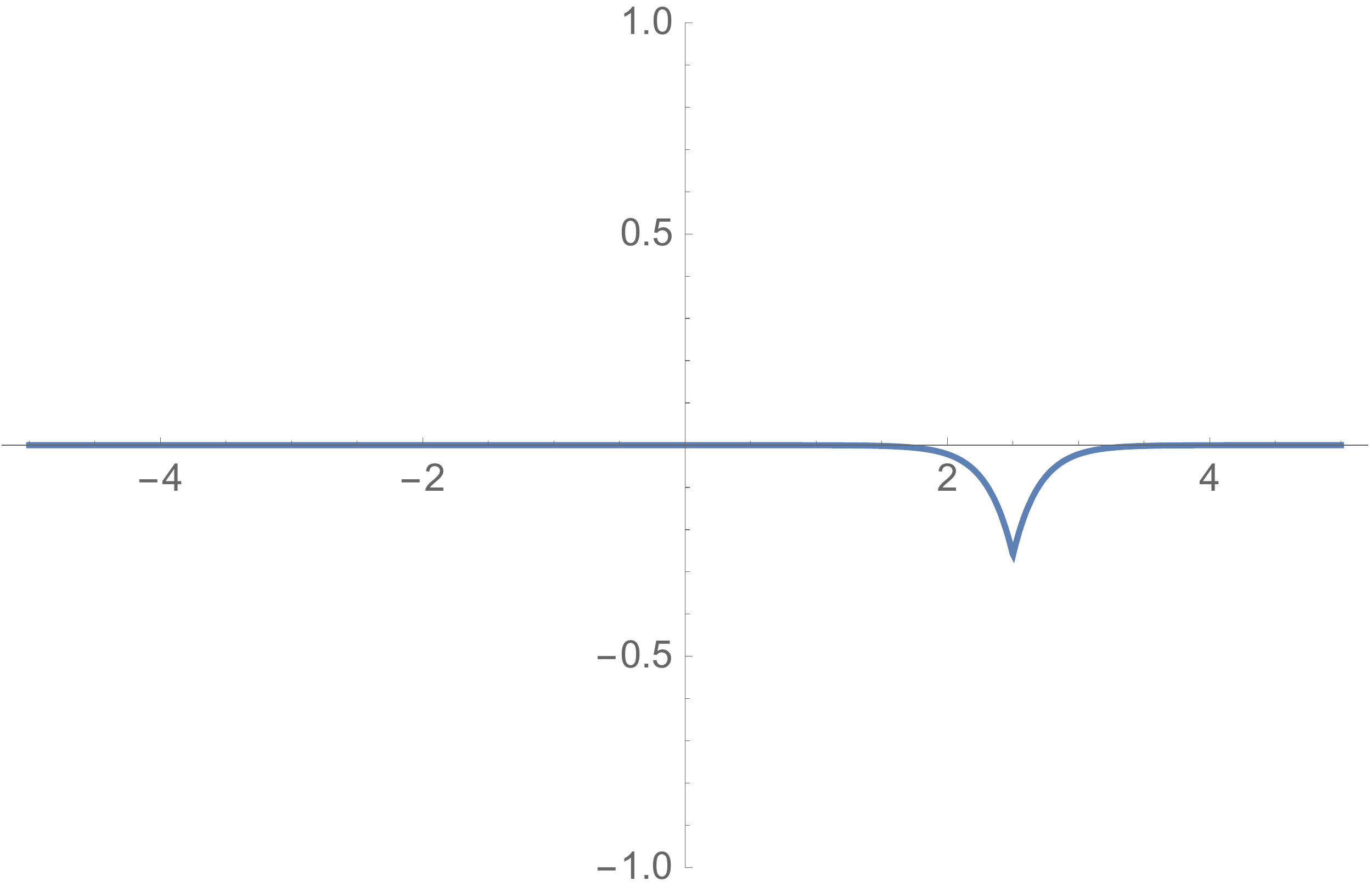}
\end{center}
\caption{\label{fond ecc} 
{\footnotesize Plot of $\phi_0(x)$ and $\phi_1(x)$ for $\gamma_{2}=-10$, $\gamma_{1}=-150$, $2a=5$.}}
\end{figure}

\noindent
As a consequence the product $|\phi_0(x) \,\,\phi_1(x)|$ turns out to be small everywhere and any periodic cancellation in (\ref{dens}) becomes impossible. No beating phenomenon will occur in this cases, as will be shown by numerical computations.

\noindent
It should be expected that the asymmetry due to the non-linearity will produce a similar behavior on time scales depending on the initial condition and on the strength of the nonlinearity. 



\subsection{Nonlinear point interactions} \label{SEC23} 
A detailed analytical study of the non-linear case $\sigma >0$ (which is no longer explicitly solvable) can be found in  (\cite{Adami:1999tk,Adami:2001bt}). The authors obtained general results about existence of solutions either local or global in time and proved existence of blow up solutions for $\sigma \geq 1$. In this section we briefly review the results that will be relevant for our work. 

\noindent
In Section \ref{SEC3}  we present the numerical simulation results for the evolution of a beating state, i.e., an initial state giving rise in the linear case to a beating motion of the particle, namely
\be \label{NL_IC} 
\psi_0 (x):=\alpha \,\phi_{f}(x)+ \beta \, \phi_{e}(x)\,, \quad \alpha,\beta \in \RR\,.
\ee
Our aim is to study how the nonlinearity influences the beating phenomenon. As we already mentioned we expect that even if the initial condition is almost-symmetric, the nonlinearity will have the effect of braking the symmetry. 


Let us go back to the general problem (\ref{SCH}) with initial conditions 
\begin{equation}
 \psi(0,x)=\psi_0(x).
\end{equation}
which we will investigate in the integral form (\ref{VOLT}). 

\noindent From (\cite{Adami:2001bt}, Theorem 6) we know that, if $\sigma < 1$ and one chooses an initial data $\psi_0 \in H^{1}(\mathbb{R})$, then the Cauchy problem has a unique solution which is global in time. Moreover in  (\cite{Adami:2001bt}, Theorem 23) it is proved that if $\gamma<0$ and $\sigma\geqslant 1$ then there exist initial data such that the solutions of the Cauchy problem  blow-up in finite time.

\noindent
For convenience of the reader we re-write below the charge equation


\be \label{VOLT2}
\left\{
\begin{array}{l} 
\ds q_1(t)+{\gamma \over 2} \sqrt{{\imath} \over  \pi} \, \int_0^t  { q_1(s)\, |q_1(s)|^{2 \sigma} \over \sqrt{t-s}}\, ds + {\gamma \over 2} \sqrt{{\imath} \over  \pi} \, \int_0^t  { q_2(s)\, |q_2(s)|^{2 \sigma} \over \sqrt{t-s}} \, e^{{\imath} { a^2 \over {t-s}}}\, ds\\[3mm]
\ds \hspace{10cm} = ({\mathcal U}(t)\, \psi_{0})(-a)\,,\\[3mm]
\ds q_2(t)+{\gamma \over 2} \sqrt{{\imath} \over  \pi} \, \int_0^t  { q_2(s)\, |q_2(s)|^{2 \sigma} \over \sqrt{t-s}}\, ds + {\gamma \over 2} \sqrt{{\imath} \over  \pi} \, \int_0^t  { q_1(s)\, |q_1(s)|^{2 \sigma} \over \sqrt{t-s}} \, e^{{\imath} { a^2 \over {t-s}}}\, ds\\[3mm]
\ds \hspace{10cm} = ({\mathcal U}(t)\, \psi_{0})(a)\,.
\end{array}
\right.
\ee
Next section is devoted to test the effectiveness of the integral form  (\ref{VOLT2}) of the evolution equations to find numerical solutions of the Cauchy problem (\ref{SCH}). 



\section{The numerical discretization of the Volterra-system} \label{SEC3} 
Let us come now to the numerical part of this work, namely the discretization and later on simulation of the  Volterra-system \eqref{VOLT2}, in order to investigate the delicate phenomenon of beating. Linear (symmetric and asymmetric) as well as nonlinear cases will be treated, starting from an initial condition under one of the forms\
\be \label{ICC}
\psi^{L}_{beat,0}(x):=\alpha\,\phi_f(x)+\beta\, \phi_e(x)\,, \quad \psi_{asy,0}(x):= \alpha\,\phi_{0}(x)+\beta\, \phi_{1}(x) \,,
\ee
with some given constants $\alpha,\beta \in \RR$ and $(\phi_f,\phi_e)$ resp. $(\phi_0,\phi_1)$ defined in \eqref{PF}-\eqref{PE} resp. \eqref{eigenfunction0}-\eqref{eigenfunction1}.\\

The discretization of the Volterra-system \eqref{VOLT2} passes through the discretization of two different kind of integrals, an Abel-integral, which is of the form
\be \label{AB}
{\mathcal Ab}(t):=\int_0^t  {g(s) \over \sqrt{t-s}}\, ds\,,
\ee
and a highly-oscillating integral of the form
\be \label{HO}
{\mathcal Ho}(t):=\int_0^t  {g(s) \over \sqrt{t-s}} \, e^{{{\mathbf i}} { a^2 \over {t-s}}}\, ds\,.
\ee
Besides, the free Schr\"odinger equation 
\be \label{SCH1D}
\left\{
\begin{array}{l}
{\mathbf i} \frac{\partial}{\partial t} \vartheta= -  \frac{\partial^2}{\partial x^2} \vartheta \,, \quad \forall (t,x) \in \RR^+ \times \RR\,,\\[3mm]
\ds \vartheta(0,x)=\psi_{0}(x)\,, \quad \forall x \in \RR\,.
\end{array}
\right.
\ee
has to be solved to compute the right hand side of the Volterra system, {\it i.e.} $({\mathcal U}(t)\, \psi_{0})(\pm a)$, and one has also to take care of the non-linearity, which will be treated iteratively by means of a linearization. The treatment of all these four steps shall be presented in the following subsections.\\

For the numerics we shall consider the truncated time-space domain $[0,T] \times [-L_x,L_x]$ and impose periodic boundary conditions in space. We shall furthermore fix a homogeneous discretization of this domain, defined as
$$
\begin{array}{llll}
&0 = t_1 < \cdots < t_l < \cdots < t_K=T\,, &\,\,\,t_l:=(l-1) \Delta t\,, & \,\,\,\Delta t := T/(K-1);\\[3mm]
&-L_x = x_1 < \cdots < x_i < \cdots < x_N=L_x\,, &\,\,\, x_i:=-L_x+(i-1) \Delta x\,, &\,\,\,\Delta x := 2L_x/(N-1)\,.
\end{array}
$$
\medskip
\subsection{The free Schr\"odinger evolution} \label{SEC_SCH}
We shall present now two different resolutions of the Schr\"odinger equation \eqref{SCH1D}, a numerical resolution via the Fast Fourier Transform (fft,ifft) assuming periodic boundary conditions in space and an analytic, explicit resolution by means of the continuous Fourier Transform and based on the specific initial condition we choose.\\

\noindent The numerical resolution starts from the partial Fourier-Transform (in space) of \eqref{SCH1D}
$$
\left\{
\begin{array}{l}
\frac{d}{d t} \hat{\theta}_k(t)= - {{\mathbf i}}\, k^2 \hat{\theta}_k(t) \,, \quad \forall k \in \ZZ\,, \quad \forall t \in \RR^+\,,\\[3mm]
\ds \hat{\theta}_k(0)=\hat{\psi}_{0,k}\,, \quad \forall k \in \ZZ\,,
\end{array}
\right.
$$
where
$$
\hat{\psi}_{0,k}:= { 1 \over 2 L_x} \int_{-L_x}^{L_x} \psi_{0}(x) e^{-{{\mathbf i}} \omega\, x \, k}\, dx\,, \quad \omega:={ \pi \over L_x}\,,
$$
and hence
\be \label{Fr_io}
\hat{\theta}_k(t)=e^{-{{\mathbf i}} \,k^2\,t}\, \hat{\theta}_k(0)\,, \quad \forall (t,k)\in \RR^+ \times \ZZ\,.
\ee
Remark that we supposed here periodic boundary conditions in the truncated space-domain $[-L_x,L_x]$, where the appearance of the discrete Fourier-variable $k \in \ZZ$. Using the fft- as well as ifft-algorithms permits hence to get from \eqref{Fr_io} a numerical approximation of the solution $\vartheta(t,x)$ of the free Schr\"odinger equation \eqref{SCH1D}.\\

Analytically, we shall situate us in the whole space $\RR$ and shall perform the same steps explicitly, taking advantage of the initial condition, which has the form
\be \label{IC}
\psi_{0}(x):=\alpha\, \phi_f(x) + \beta\, \phi_e(x)\,, \quad \forall x \in \RR\,, \quad \alpha,\beta \in \RR\,,
\ee
where we recall that (see \eqref{GG}, \eqref{PF}, \eqref{PE})
$$
\phi_f(x)=N_f \left[ G^{\lambda_f} (x+a) + G^{\lambda_f} (x-a)\right]\,, \quad \phi_e(x)=N_e \left[ G^{\lambda_e}  (x+a) - G^{\lambda_e}  (x-a)\right]\,.
$$
Thus one has with the definition of the Fourier-transform and its inverse
$$
\hat{\phi}(\nu):= {1 \over \sqrt{2 \pi}} \int_{-\infty}^\infty \phi(x)\, e^{-{{\mathbf i}}\, x\, \nu} dx\,, \quad \phi(x)={1 \over \sqrt{2 \pi}} \int_{-\infty}^\infty \hat{\phi}(\nu) \, e^{{{\mathbf i}}\, x\, \nu} d\nu\,,
$$
that
$$
\hat{\psi}_0(\nu)=\alpha \, \hat{\phi}_f(\nu)+ \beta\, \hat{\phi}_e(\nu) \quad \Rightarrow \quad \hat{\vartheta}(t,\nu)=\alpha \, \hat{\phi}_f(\nu)\, e^{-{{\mathbf i}}\, \nu^2\, t}+ \beta\, \hat{\phi}_e(\nu)\,e^{-{{\mathbf i}}\, \nu^2\, t} \quad (t,\nu) \in \RR^+ \times \RR\,.
$$
Let us now compute explicitly the Fourier transform of $\phi_{0}$ and $\phi_1$ and finally the inverse Fourier transform of $\hat{\vartheta}(t,\nu)$. For this, remark that one has
$$
\hat{G}^\lambda(\nu)={1 \over \sqrt{2\, \pi}}\, { 1 \over \lambda+ \nu^2}\,, \quad \forall \nu \in \RR\,,
$$
leading to
$$
\hat{\phi}_f(\nu)={2 N_f \over \sqrt{2\, \pi}}\, { \cos(\nu\, a) \over \lambda_f+ \nu^2}\,, \quad \hat{\phi}_e(\nu)=- {2 \,{{\mathbf i}}\, N_e \over \sqrt{2\, \pi}} \,{ \sin(\nu\, a) \over \lambda_e+ \nu^2}\,.
$$
Now, in the aim to resolve numerically the Volterra-system \eqref{VOLT2}, one needs only to compute the solution of \eqref{SCH1D} in the points $y_{1,2}=\pm a$, which means
$$
\vartheta(t,-a)={\alpha\,N_f \over 2 \pi} \int_{-\infty}^\infty {1+e^{-2\,{{\mathbf i}}\, a\, \nu } \over \lambda_f+\nu^2}\,e^{-{{\mathbf i}}\, \nu^2\, t} \, d\nu + {\beta\,N_e \over 2 \pi} \int_{-\infty}^\infty {1-e^{-2\,{{\mathbf i}}\, a\, \nu } \over \lambda_e+\nu^2}\,e^{-{{\mathbf i}}\, \nu^2\, t} \, d\nu\,.
$$
$$
\vartheta(t,a)={\alpha\,N_f \over 2 \pi} \int_{-\infty}^\infty {1+e^{2\,{{\mathbf i}}\, a\, \nu } \over \lambda_f+\nu^2}\,e^{-{{\mathbf i}}\, \nu^2\, t} \, d\nu - {\beta\,N_e \over 2 \pi} \int_{-\infty}^\infty {1-e^{2\,{{\mathbf i}}\, a\, \nu } \over \lambda_e+\nu^2}\,e^{-{{\mathbf i}}\, \nu^2\, t} \, d\nu\,,
$$
To compute these two integrals, we shall take advantage of the following two formulae
$$
I^\lambda_A:=\int_{-\infty}^{\infty} {1 \over \lambda+\nu^2}\,e^{-{{\mathbf i}}\, \nu^2\, t} \, d\nu={\pi \over \sqrt{\lambda}}\, e^{{{\mathbf i}}\, \lambda\,t} \left[1- \textrm{erf}(\sqrt{{{\mathbf i}}\, \lambda \, t}) \right]\,,
$$
$$
I^\lambda_B:=\int_{-\infty}^{\infty} { e^{\pm 2\,{{\mathbf i}}\, a\, \nu }\over \lambda+\nu^2}\,e^{-{{\mathbf i}}\, \nu^2\, t} \, d\nu=\int_{-\infty}^{\infty} { \cos(2\, a\,\nu )\over \lambda+\nu^2}\,e^{-{{\mathbf i}}\, \nu^2\, t} \, d\nu\,,
$$
where $\textrm{erf}(\cdot)$ is the so-called error-function, defined by
$$
\textrm{erf}(x):={2 \over \sqrt{\pi}} \int_0^x e^{-t^2}\, dt\,.
$$
After some straightforward computations, one gets
$$
I^\lambda_B={\pi \over 2\, \sqrt{\lambda}}\, e^{{{\mathbf i}}\, \lambda\,t}\, \left\{  e^{2\, \sqrt{\lambda}\,a}\, \left[ 1-\textrm{erf}\left( \sqrt{{{\mathbf i}} \, \lambda\, t} + { a \over \sqrt{{{\mathbf i}} \, t}}\right)\right] + e^{-2\, \sqrt{\lambda}\,a}\,\left[ 1-\textrm{erf}\left( \sqrt{{{\mathbf i}} \, \lambda\, t} - { a \over \sqrt{{{\mathbf i}} \, t}}\right)\right]\right\}\,.
$$
With the two expressions $I^\lambda_A$ and $I^\lambda_B$ one has now
$$
({\mathcal U}(t)\, \psi_{0})(-a)=\vartheta(t,-a)={\alpha\,N_f \over 2 \pi} \, \left[ I_A^{\lambda_f}+I_B^{\lambda_f}\right] -{\beta\,N_e \over 2 \pi}\, \left[ I_A^{\lambda_e}+I_B^{\lambda_e}\right]\,,
$$
$$
({\mathcal U}(t)\, \psi_{0})(a)=\vartheta(t,a)={\alpha\,N_f \over 2 \pi} \, \left[ I_A^{\lambda_f}+I_B^{\lambda_f}\right] +{\beta\,N_e \over 2 \pi}\, \left[ I_A^{\lambda_e}+I_B^{\lambda_e}\right]\,,
$$
which permits to have the right-hand side of the Volterra-system \eqref{VOLT2} analytically.\\
Let us observe that the same computations hold also for the asymmetric initial condition \eqref{AsyINIT} with $(\phi_f,\phi_e)$ replaced by $(\phi_0,\phi_1)$, as well as $(N_f,N_e)$ by $(N_0,N_1)$.
\subsection{The Abel integral} \label{SEC_ABEL}
Let us now present a discretization of an Abel-integral of the form \eqref{AB}, based on a Gaussian quadrature. The time interval $[0,T]$ is discretized in a homogeneous manner, as proposed above, such that one can now approximate ${\mathcal Ab}(t_l)$ for $l=1,\cdots ,K$ as follows
$$
{\mathcal Ab}(t_l)=\int_0^{t_l} {g(s) \over \sqrt{t_l-s}}\, ds= \sum_{k=1}^{l-1} \int_{t_k}^{t_{k+1}} {g(s) \over \sqrt{t_l-s}}\, ds= \sum_{k=1}^{l-1} \sqrt{\Delta t} \int_{0}^{1} {g(t_k+ \xi \Delta t) \over \sqrt{l-k-\xi}}\, d\xi\,.
$$
Now, introducing the notation
$$
r_k^{(l)} := l-k\,, \quad \varphi_k(\xi):= g(t_k+ \xi \Delta t)\,, \quad p_k^{(l)}(\xi):= { 1 \over \sqrt{r_k^{(l)}-\xi}}\,,
$$
we will use a Gaussian quadrature formula with one point and the weight-function $p_k^{(l)}(\xi)$ to approximate the last integral as follows
$$
\int_0^1 p_k^{(l)}(\xi) \, \varphi_k(\xi)\, d\xi = w_k^{(l)}\, \varphi_k(\xi_k^{(l)})\,,
$$
with the ``Gauss-points'' given by
\be \label{GPS}
w_k^{(l)}:= \int_0^1 { 1 \over \sqrt{r_k^{(l)}-\eta}}\, d\eta\,, \quad \xi_k^{(l)}:= { 1 \over w_k^{(l)}}\, \int_0^1 { \eta \over \sqrt{r_k^{(l)}-\eta}}\, d\eta\,.
\ee
This leads to
$$
{\mathcal Ab}(t_l)\approx \sqrt{\Delta t}\, \sum_{k=1}^{l-1}  w_k^{(l)}\, g(t_k+\Delta t\, \xi_k^{(l)})\,.
$$
As the function $g$ is known only at the grid points $t_k$, we shall linearize $g$ in the cell $[t_k,t_{k+1}]$ to find finally the approximation formula we used for the Abel-integral
\be \label{AB_D}
{\mathcal Ab}(t_l)\approx {\mathcal Ab}_{num}(t_l):= \sqrt{\Delta t}\,\sum_{k=1}^{l-1}  w_k^{(l)}\, \left[ \xi_k^{(l)}\, g_{k+1} +(1-\xi_k^{(l)})\, g_k \right]\,, \quad \forall l=1, \cdots,K\,,
\ee
where $w_k^{(l)}$ and $\xi_k^{(l)}$ are given by \eqref{GPS} and $g_k:=g(t_k)$.
Let us remark here that the function $g(s)$ is known up to the instant $t_{l-1}$, such that we have to keep in mind that there is a term in this last formula, which is unknown, {\it i.e.} $g_l$, and which has to be computed at this present step via the Volterra-system. This procedure shall be explained in subsection \ref{SEC_NON}, however let us here introduce some notation, to simplify the subsequent analysis. We shall denote
$$
{\mathcal Ab}_{num}^1(t_l):= \sqrt{\Delta t}\,\sum_{k=1}^{l-2}  w_k^{(l)}\, \left[ \xi_k^{(l)}\, g_{k+1} +(1-\xi_k^{(l)})\, g_k \right] + \sqrt{\Delta t}\,w_{l-1}^{(l)}\,(1-\xi_{l-1}^{(l)})\, g_{l-1} \,, \quad \forall l=1, \cdots,K\,,
$$
and
$$
{\mathcal Ab}_{num}^2(t_l):= \sqrt{\Delta t}\, w_{l-1}^{(l)}\,  \xi_{l-1}^{(l)}\, g_{l} \,,
$$
such that \eqref{AB_D} becomes simply
\be \label{AB_D_bis}
{\mathcal Ab}(t_l)\approx {\mathcal Ab}_{num}(t_l)= {\mathcal Ab}_{num}^1(t_l) +\sqrt{\Delta t}\, w_{l-1}^{(l)}\,  \xi_{l-1}^{(l)}\, g_{l} \,, \quad \forall l=1, \cdots,K\,.
\ee
\subsection{The Highly-oscillating integral} \label{SEC_HOSC}
Let us come now to the treatment of the highly oscillatory integral \eqref{HO}, which is the most delicate part of our numerical scheme. Indeed, as one can observe from Fig. \ref{IMA_INT}, the integrand function (here with $g \equiv 1$, $t=a=1$) is a rapidly varying function such that its integration has to be done with care. 
\begin{figure}[htbp]
\begin{center}
\includegraphics[width=7cm]{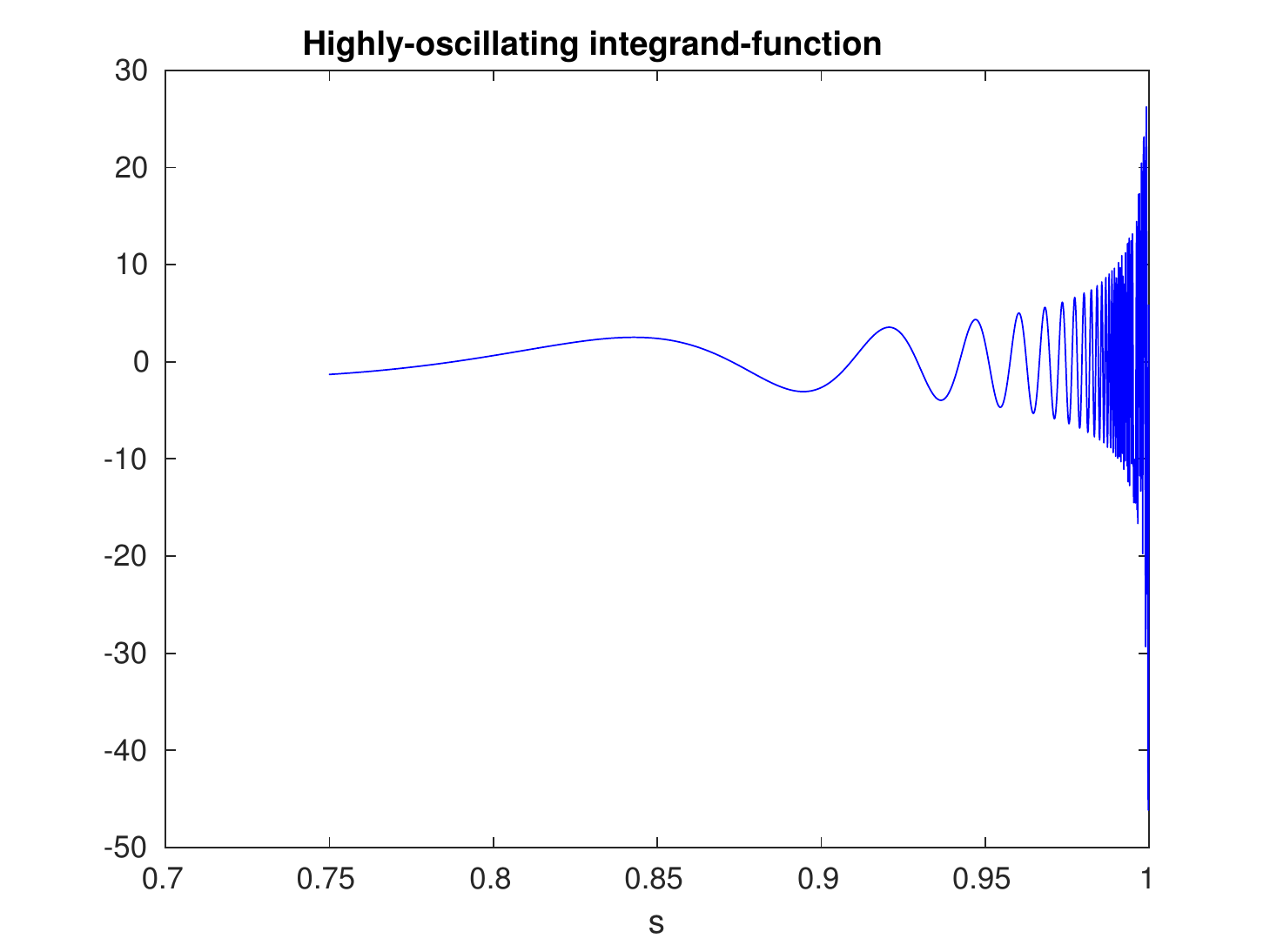}
\end{center}
\caption{\label{IMA_INT} 
{\footnotesize Evolution in time of the integrand function $h(s):={1 \over \sqrt{1-s}} \, e^{{{\mathbf i}} { 1 \over {1-s}}}$ with $s \in [0,1]$.}}
\end{figure}
We shall present here different procedures for its computation or approximation. The first procedure is more analytical and based on integral-tables \cite{Abramo}. The second one is a numerical approach and uses integration-by-parts to cope with the high oscillations.\\

The analytic procedure starts with linearizing the function $g$ in the cell $[t_k,t_{k+1}]$, 
$$
g(s)={ g_{k+1}-g_k \over \Delta t}\, (s-t_k) + g_k\,, \quad \forall s \in [t_k,t_{k+1}]\,,
$$
in order to approximate
$$
\begin{array}{lll}
\ds {\mathcal Ho}(t_l)&:=&\ds \int_0^{t_l}  {g(s) \over \sqrt{t_l-s}} \, e^{{{\mathbf i}} { a^2 \over {t_l-s}}}\, ds\\[3mm]
&\approx&\ds \sum_{k=1}^{l-1}  \left[ { g_{k+1}-g_k \over \Delta t} \int_{t_k}^{t_{k+1}} {s-t_k \over \sqrt{t_l-s}} \, e^{{{\mathbf i}} { a^2 \over {t_l-s}}}\, ds +  g_k \int_{t_k}^{t_{k+1}} {1 \over \sqrt{t_l-s}} \, e^{{{\mathbf i}} { a^2 \over {t_l-s}}}\, ds\right]\\[3mm]
&=&\ds \sum_{k=1}^{l-1} \left[ { g_{k+1}-g_k \over \Delta t}\, I^{k,l}_1 +  g_k \,I^{k,l}_2\right]\,, \quad \forall l =1, \cdots,K\,,
\end{array}
$$ 
where we denoted
$$
I^{k,l}_1:=\int_{t_k}^{t_{k+1}} {s-t_k \over \sqrt{t_l-s}} \, e^{{{\mathbf i}} { a^2 \over {t_l-s}}}\, ds\,, \qquad I^{k,l}_2:=\int_{t_k}^{t_{k+1}} {1 \over \sqrt{t_l-s}} \, e^{{{\mathbf i}} { a^2 \over {t_l-s}}}\, ds\,.
$$
The $I^{k,l}_1$-integral can be further developped as follows
$$
I^{k,l}_1=\int_{t_k}^{t_{k+1}} {s-t_l \over \sqrt{t_l-s}} \, e^{{{\mathbf i}} { a^2 \over {t_l-s}}}\, ds+ (t_l-t_k)\int_{t_k}^{t_{k+1}} {1\over \sqrt{t_l-s}} \, e^{{{\mathbf i}} { a^2 \over {t_l-s}}}\, ds= - I^{k,l}_3 + (t_l-t_k)\, I^{k,l}_2\,,
$$
where we introduced
$$
I^{k,l}_3:=\int_{t_k}^{t_{k+1}} \sqrt{t_l-s} \,\, e^{{{\mathbf i}} { a^2 \over {t_l-s}}}\, ds\,.
$$
The integrals $I^{k,l}_2$ and $I^{k,l}_3$ have now explicit expressions. Indeed, one can find, using \cite{Abramo}, that
\be \label{I2}
I^{k,l}_2= 2 \int_{D_{k+1}}^{D_k} e^{{{\mathbf i}}\, a^2/\xi^2}\, d\xi=2 \, \left[ \sqrt{-{{\mathbf i}}\, \pi}\, a\, \textrm{erf} \left( {\sqrt{-{{\mathbf i}}}\, a \over \xi} \right) \right]_{D_{k+1}}^{D_{k}}\,, \quad D_k:=\sqrt{t_l-t_{k}}\,;
\ee
and
\be \label{I3}
I^{k,l}_3=2 \int_{D_{k+1}}^{D_{k}} \xi^2\, e^{{{\mathbf i}}\, a^2/\xi^2}\, d\xi={2\over 3} \, \left[\xi^3   \right]_{D_{k+1}}^{D_{k}} + 4 \sqrt{ {{\mathbf i}} \pi}\, (T_{D_{k}}-T_{D_{k+1}})\,,
\ee
with
$$
T_{D}=\left[ {\xi^3 \over 3} \, \textrm{erf} \left( {\sqrt{-{{\mathbf i}}}\over D}\, \xi \right) + e^{\frac{{\mathbf i}}{D^2}\, \xi^2 }\, \frac{D}{3 \sqrt{{-{\mathbf i}}\, \pi}} \left( \xi^2 + {{\mathbf i}}\,  D^2\right) \right]_0^a\,, \quad \textrm{for}\,\,\, D=D_k\,, D_{k+1}\,.
$$
Using now these explicit formulae \eqref{I2}-\eqref{I3} we get an approximate formula for the highly-oscillating intergral ${\mathcal Ho}$, {\it i.e.}
\be \label{HO1}
{\mathcal Ho}(t_l) \approx \sum_{k=1}^{l-1} \left[ { g_{k+1}-g_k \over \Delta t}\, (-I^{k,l}_3 + (t_l-t_k) I^{k,l}_2) +  g_k \,I^{k,l}_2\right]\,, \quad \forall l =1, \cdots,K\,,
\ee
This formula is quasi-analytical, and is based on the linearization of the function $g$. This linearization is possible, if the function $g$ itself is not highly-oscillating. We remark here also that \eqref{HO1} involves the still unknown value $g_l$.\\

A second idea can be used to approximate these highly oscillating integrals, based more on a numerical discretization. Let us start from
$$
\begin{array}{lll}
\ds {\mathcal Ho}(t_l)&=&\ds \int_0^{t_l}  {g(s) \over \sqrt{t_l-s}} \, e^{{{\mathbf i}} { a^2 \over {t_l-s}}}\, ds\\[3mm]
&=&\ds \sum_{k=1}^{N^{(l)}_{it}-1} \int_{t_k}^{t_{k+1}}  {g(s) \over \sqrt{t_l-s}} \, e^{{{\mathbf i}} { a^2 \over {t_l-s}}}\, ds+ \sum_{k=N^{(l)}_{it}}^{l-2} \int_{t_k}^{t_{k+1}}  {g(s) \over \sqrt{t_l-s}} \, e^{{{\mathbf i}} { a^2 \over {t_l-s}}}\, ds+\int_{t_{l-1}}^{t_{l}}  {g(s) \over \sqrt{t_l-s}} \, e^{{{\mathbf i}} { a^2 \over {t_l-s}}}\, ds\\[3mm]
& =:& \ds I_{H1}+I_{H2}+I_{H3}\,, \quad \forall l=1, \cdots,K\,.
\end{array}
$$
This decomposition follows the evolution of the integrand-function, meaning that the index $N^{(l)}_{it} \in [1,l-1] \subset \NN$ will delimitate the regions of smooth evolution resp. rapid variation and shall permit two different treatments of the integrals. This index is chosen in our case in the following manner:\\
The highly oscillating function $e^{{{\mathbf i}}{a^2 \over t_l-s}}$ has a period which diminishes monotonically as $s \rightarrow t_l$. The extrema of this function are localized at the points $s_j:=t_l-{a^2 \over j\, \pi}$, $j \in \NN$ and $s_j \rightarrow_{j \rightarrow \infty} t_l$. A function is smooth in our sens, if between two extrema we have at least $5$ time-steps. Hence, letting $J \in \NN$ being the index, such that $5 \star  \Delta t \sim s_{J+1}-s_J={a^2 \over \pi} \left[ \frac{1}{J}-\frac{1}{J+1} \right] \sim { a^2 \over \pi\, J^2}$, we define $N^{(l)}_{it}$ such that $t_{N^{(l)}_{it}} < s_J < t_{N^{(l)}_{it}+1}$.\\
Now for $k<N^{(l)}_{it}$ the integrand function is not so oscillating, and a standard quadrature-method (for example rectangle or trapez-method) can be used to approximate $I_{H1}$. In particular, using the trapez-method leads to 
$$
I_{H1}\approx \sum_{k=1}^{N^{(l)}_{it}-1} {\Delta t \over 2} \left[ {g_{k+1} \over \sqrt{t_l-t_{k+1}}} \, e^{{{\mathbf i}} { a^2 \over {t_l-t_{k+1}}}} +{g_{k} \over \sqrt{t_l-t_{k}}} \, e^{{{\mathbf i}} { a^2 \over {t_l-t_{k}}}} \right]\,.
$$
For $k\ge N^{(l)}_{it}$ the integrand function is becoming too oscillating to use any more standard quadrature methods, such that we shall rather make use of an integration-by-parts (IPP) technique, {\it i.e.}
$$
\begin{array}{lll}
\ds \int_{t_k}^{t_{k+1}}  {g(s) \over \sqrt{t_l-s}} \, e^{{{\mathbf i}} { a^2 \over {t_l-s}}}\, ds &=& \ds \left( - {{{\mathbf i}} \over a^2} \right) \, \int_{t_k}^{t_{k+1}}  g(s)\, (t_l-s)^{3/2} \,\left(  {{{\mathbf i}} a^2 \over (t_l-s)^2 }\, e^{{{\mathbf i}} { a^2 \over {t_l-s}}} \right)\, ds\\[4mm]
& = &\ds {{{\mathbf i}} \over a^2}  \, \int_{t_k}^{t_{k+1}} \left[ g(s)\, (t_l-s)^{3/2}\right]'\, e^{{{\mathbf i}} { a^2 \over {t_l-s}}}\, ds -{{{\mathbf i}} \over a^2} \left[g(s)\, (t_l-s)^{3/2}\, e^{{{\mathbf i}} { a^2 \over {t_l-s}}} \right]_{t_k}^{t_{k+1}}\\[4mm]
&\approx&\ds \left( - {{{\mathbf i}} \over a^2} \right) \, \left[ g_{k+1}\, (t_l-t_{k+1})^{3/2}\, e^{{{\mathbf i}} { a^2 \over {t_l-t_{k+1}}}}-g_{k}\, (t_l-t_{k})^{3/2}\, e^{{{\mathbf i}} { a^2 \over {t_l-t_{k}}}} \right]\,;\\[4mm]
\ds \int_{t_{l-1}}^{t_l}  {g(s) \over \sqrt{t_l-s}} \, e^{{{\mathbf i}} { a^2 \over {t_l-s}}}\, ds &\approx& \ds {{{\mathbf i}} \over a^2}\, g_{l-1}\, (\Delta t)^{3/2}\, e^{{{\mathbf i}}  a^2/ \Delta t}\,.
\end{array}
$$
Using these formulae, and remarking the telescopic summation, one gets immediately
$$
I_{H2}+I_{H3}\approx{{{\mathbf i}} \over a^2} \,g_{N^{(l)}_{it}}\, (t_l-t_{N^{(l)}_{it}})^{3/2}\, e^{{{\mathbf i}} { a^2 \over {t_l-t_{N^{(l)}_{it}}}}}\,.
$$
Hence, we get altogether
\be \label{HO_bis}
{\mathcal Ho}_{num}(t_l)=\sum_{k=1}^{N^{(l)}_{it}-1} {\Delta t \over 2} \left[ {g_{k+1} \over \sqrt{t_l-t_{k+1}}} \, e^{{{\mathbf i}} { a^2 \over {t_l-t_{k+1}}}} +{g_{k} \over \sqrt{t_l-t_{k}}} \, e^{{{\mathbf i}} { a^2 \over {t_l-t_{k}}}} \right]+{{{\mathbf i}} \over a^2} \,g_{N^{(l)}_{it}}\, (t_l-t_{N^{(l)}_{it}})^{3/2}\, e^{{{\mathbf i}} { a^2 \over {t_l-t_{N^{(l)}_{it}}}}}\,.
\ee
Remark at this point that in this case, we do not need $g_l$ for the computation of ${\mathcal Ho}(t_l)$.
\subsection{The Non-linearity} \label{SEC_NON}
The non-linearity is treated iteratively, by linearizing the non-linear term. 
To explain this procedure, let us first summarize what we performed up to now in the discretization of the Volterra-system \eqref{VOLT2}. Denoting for simplicity the constant $\kappa:={\gamma \over 2} \sqrt{{{\mathbf i}} \over  \pi}$ and using the approximations \eqref{AB_D_bis} as well as \eqref{HO_bis}, we have for $l=1, \cdots,K$
\be \label{VOLT_BIS}
\left\{
\begin{array}{l}
\ds q_1^l+\kappa \,\sqrt{\Delta t}\, w_{l-1}^{(l)}\,  \xi_{l-1}^{(l)}\, q_1^l\, |q_1^l|^{2 \sigma}+\kappa \, {\mathcal Ab}_{num}^{1,-}(t_l)  + \kappa \, {\mathcal Ho}_{num}^+(t_l)= \vartheta(t_l,-a)\\[3mm]
\ds q_2^l+\kappa \,\sqrt{\Delta t}\, w_{l-1}^{(l)}\,  \xi_{l-1}^{(l)}\, q_2^l\, |q_2^l|^{2 \sigma}+\kappa \, {\mathcal Ab}_{num}^{1,+}(t_l)  + \kappa \, {\mathcal Ho}_{num}^-(t_l)= \vartheta(t_l,a)\,,

\end{array}
\right.
\ee
where in the Abel and highly-oscillating terms we introduced a sign $\pm$ in order to underline which function $g(s)=q_{1,2}(s)\, |q_{1,2}(s)|^{2 \sigma}$ they involve, in particular the one corresponding to $y_1=-a$ or to $y_2=+a$.\\
The resolution of the non-linear system \eqref{VOLT_BIS} consists now in introducing the linearization-sequence $\{q_{1,2}^{l,n}\}_{n \in \NN}$ as follows
$$
q_{1,2}^{l,0}:=q_{1,2}^{l-1}\,,
$$
and where the terms $q_{1,2}^{l,n}$ are solution for $n\ge 1$ of the linearized Volterra-system
\be \label{VOLT_LIN}
\left\{
\begin{array}{l}
\ds q_1^{l,n}+\kappa \,\sqrt{\Delta t}\, w_{l-1}^{(l)}\,  \xi_{l-1}^{(l)}\, q_1^{l,n}\, |q_1^{l,n-1}|^{2 \sigma}+\kappa \, {\mathcal Ab}_{num}^{1,-}(t_l)  + \kappa \, {\mathcal Ho}_{num}^+(t_l)= \vartheta(t_l,-a)\\[3mm]
\ds q_2^{l,n}+\kappa \,\sqrt{\Delta t}\, w_{l-1}^{(l)}\,  \xi_{l-1}^{(l)}\, q_2^{l,n}\, |q_2^{l,n-1}|^{2 \sigma}+\kappa \, {\mathcal Ab}_{num}^{1,+}(t_l)  + \kappa \, {\mathcal Ho}_{num}^-(t_l)= \vartheta(t_l,a)\,.
\end{array}
\right.
\ee
This procedure is stopped at $k=M$, either when two subsequent iterations do not vary any more, meaning $|q_{1,2}^{l,n}-q_{1,2}^{l,n-1}| < 10^{-3}$, or when a maximal number of interations, as $k=10$ is reached, and one defines finally
$$
q_{1,2}^{l}:=q_{1,2}^{l,M}\,.
$$
Solving now the system \eqref{VOLT_LIN} permits us to get a numerical approximation of the solution to the Volterra-system \eqref{VOLT2} and in the next section we shall present the simulations based on the just presented scheme.
\section{Numerical simulation of the beating phenomenon} \label{CH4_SEC4} 
Let us present in this section the numerical results obtained with the scheme presented in Section \ref{SEC3}. First we shall investigate the symmetric and asymmetric linear case and compare the obtained results with the exact solutions in order to validate the code. A particular attention is paid to the asymmetric linear case, which does not allow for a beating motion of the particle, the initial symmetry being destroyed. Secondly we shall pass to the non-linear simulations and study the destruction of the beating due to the manifestation of the non-linearity.
\subsection{The symmetric linear case} \label{LIN}
In this section we set $\sigma=0$ and consider the linear Volterra-system \eqref{VOLT2}-\eqref{SCH1D} associated with the initial condition given in \eqref{IC}, namely
\be \label{ICB}
\psi^{L}_{beat,0}(x):=\alpha\,\phi_f(x)+\beta\, \phi_e(x)\,,  \quad \alpha,\beta\in \RR\,,
\ee
and corresponding exact solution
\be \label{CARICHE}
\left\{
\begin{array}{l}
q_1(t)=\alpha\, \phi_f(-a)\, e^{{{\mathbf i}}\,\lambda_f\,t} +\beta\, \phi_e(-a)\, e^{{{\mathbf i}}\,\lambda_e\,t}\,, \\[3mm]
q_2(t)=\alpha\, \phi_f(a)\, e^{{{\mathbf i}}\,\lambda_f\,t} +\beta\, \phi_e(a)\, e^{{{\mathbf i}}\,\lambda_e\,t}\,, 
\end{array}
\right.
\quad \forall t \in \RR^+\,.
\ee
In the following linear tests, we performed the simulations with the parameters
\be \label{PARA}
a=3\,, \quad \alpha=\sqrt{0.01}\,, \,\,\, \beta=\sqrt{0.99}\,, \quad \gamma=-0.5\,.
\ee

\noindent Figure \ref{IMA0} presents on the left the time-evolution of the numerical solutions of the Volterra-system \eqref{VOLT2}-\eqref{SCH1D}, associated to the parameters presented above, and on the right the relative error between the exact solution and the numerical solution. One can firstly observe the so-called beating motion of the system between the two ``stable'' configurations, which correspond to the first two energy states of the nitrogen atom.\\ 
Secondly, one remarks also a nice overlap of the numerical with the exact solutions. This overlap begins to deteriorate in time, effect which comes from the accumulation of the numerical errors, arising during the approximations we perform in the simulation. These linear tests permitted us to validate the linear version of our code. 
\begin{figure}[htbp]
\begin{center}
\includegraphics[width=7cm]{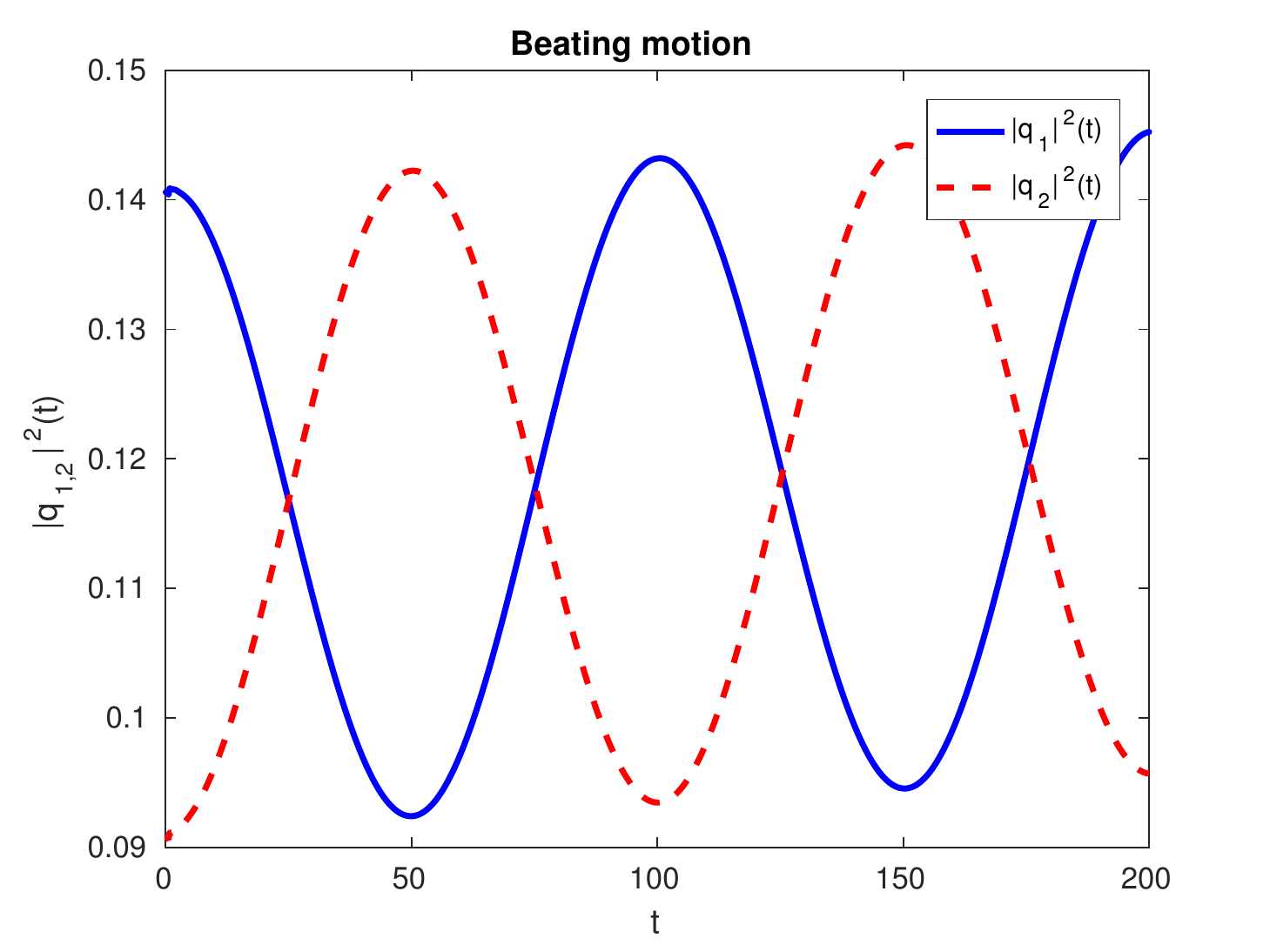}\hfill
\includegraphics[width=7cm]{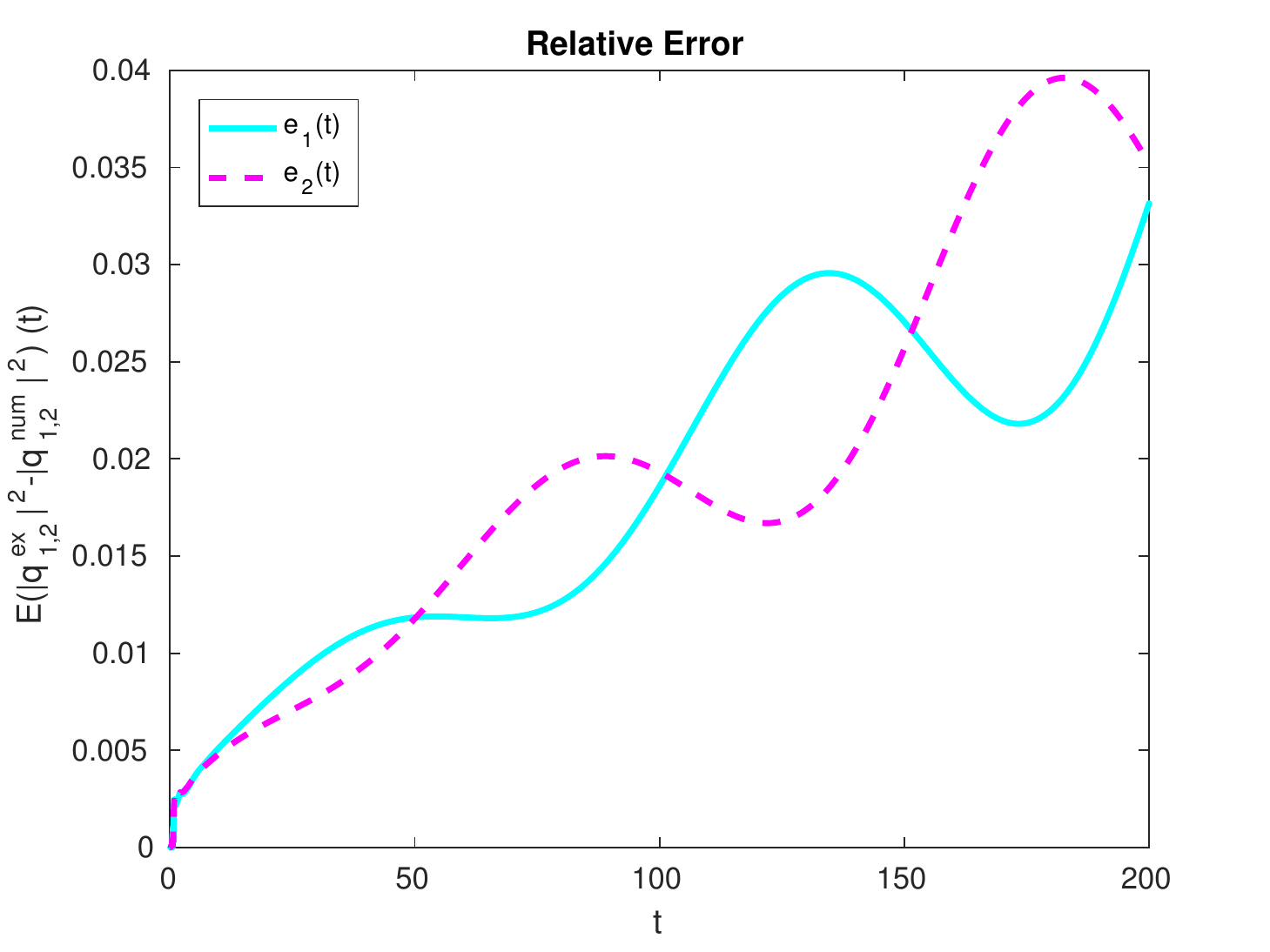}
\end{center}
\caption{\label{IMA0} 
{\footnotesize The beating effect. Left: Evolution in time of the numerical solutions $|q_1|^2(t)$ resp. $|q_2|^2(t)$. Right: Relative error: $ \textrm{abs}\left[|q_{1,2}^{ex}|^2-|q_{1,2}^{num}|^2 \right] / ||q_{1,2}^{ex}||_\infty$.}}
\end{figure}

\subsection{The asymmetric linear case} \label{ASY}
In contrast to the previous case, we shall now choose an asymmetric initial condition of the form
$$
\psi_{asy,0}(x):= \alpha\,\phi_{0}(x)+\beta\, \phi_{1}(x) \,, \quad \alpha,\beta \in \RR\,,
$$
with $\phi_0$, $\phi_1$ defined in \eqref{eigenfunction0}-\eqref{eigenfunction1}. The exact solution is equally known in this case, and is given by similar a formula as \eqref{CARICHE} (see \eqref{EX_ASY}). 
Two plots are presented in Fig. \ref{IMA_ASY}, corresponding to the two sets of parameters:
\begin{itemize}
\item (A) $\gamma_1=-8$,\,\, $\gamma_2=-4$,\,\, $a=1/2$,\,\, $\alpha=\beta=1/ \sqrt{2}$\,;
\item (B) $\gamma_1=-10$,\,\, $\gamma_2=-4$,\,\, $a=5$,\,\, $\alpha=\beta=1/ \sqrt{2}$\,.
\end{itemize}
As expected, the beating motion of the nitrogen atom is completely annihilated, and this due to the asymmetric initial conditions. In Fig. \ref{IMA_ASY} (A) one observes that the particle remains with a certain probability in each potential well, without crossing the barrier by tunneling and jumping in the other well. In each well, the particle is performing a periodic motion, permitting to show that the particle is not at rest in the well. In Fig. \ref{IMA_ASY} (B) the parameters are more extreme and the particle seems even to be at rest in the two wells.
\begin{figure}[htbp]
\begin{center}
\includegraphics[width=7cm]{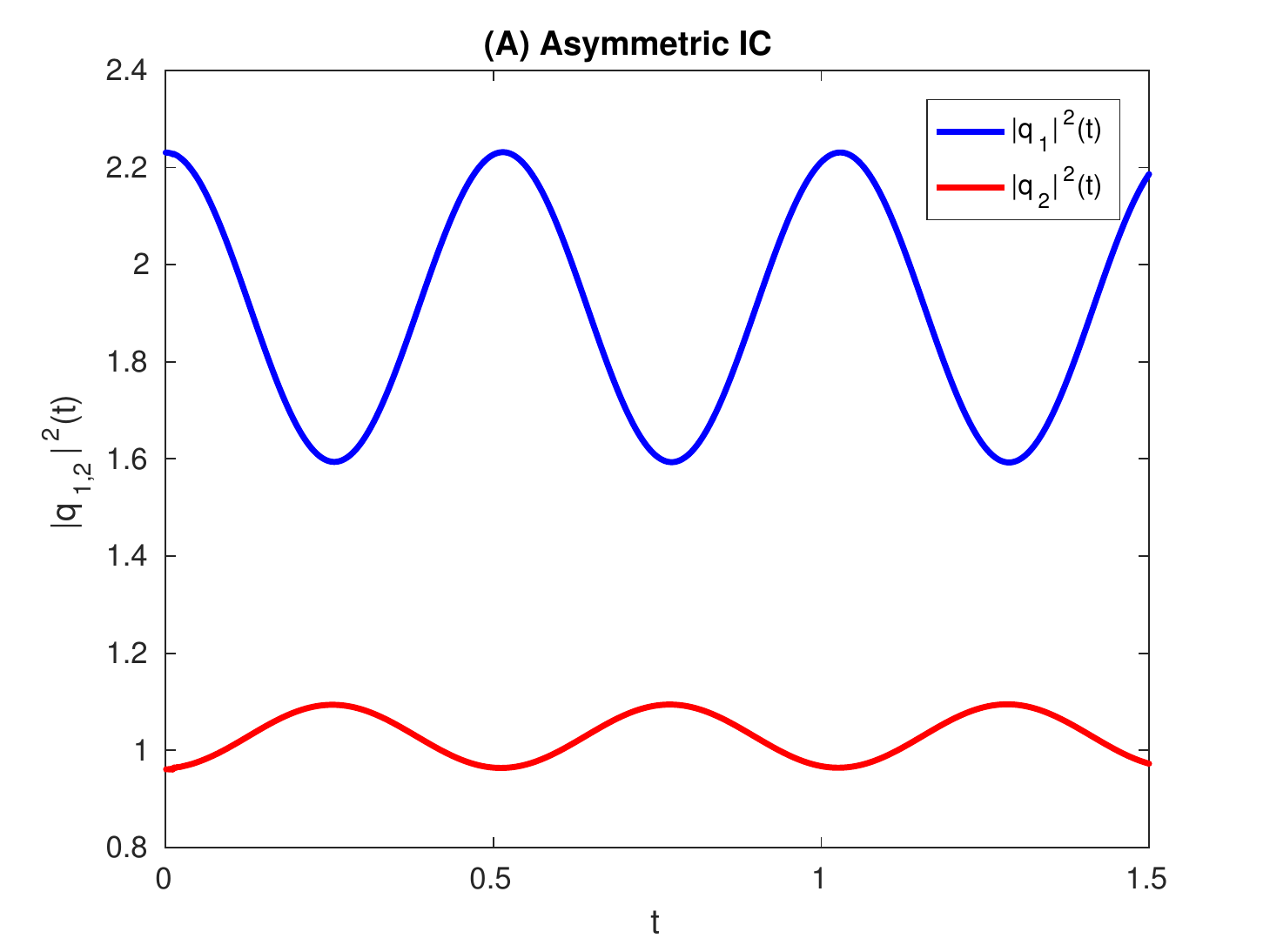}\hfill
\includegraphics[width=7cm]{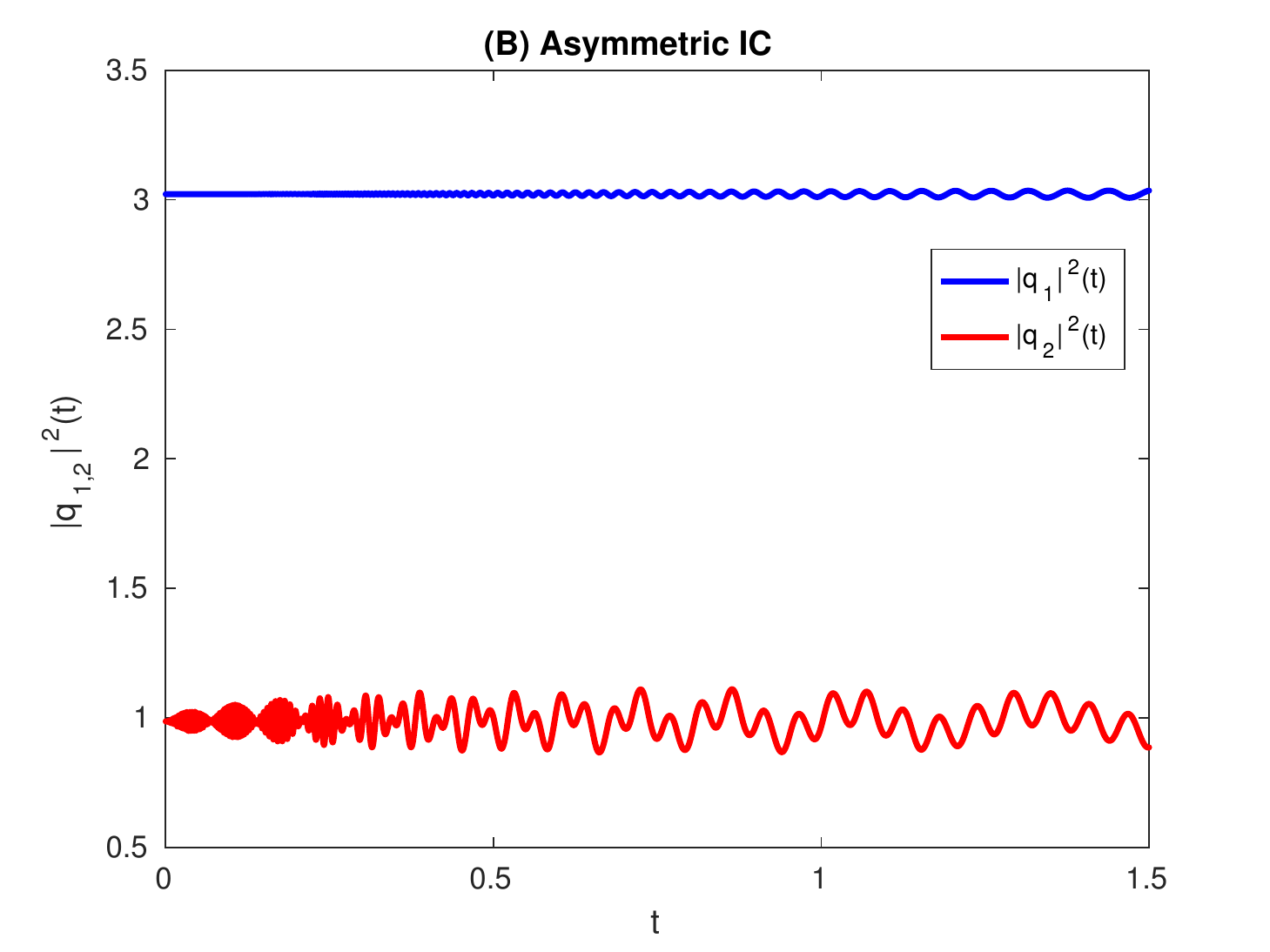}
\end{center}
\caption{\label{IMA_ASY} 
{\footnotesize The asym. lin. case. Left: Evolution in time of the numerical solutions $|q_1|^2(t)$ resp. $|q_2|^2(t)$ with the set of parameters (A). Right: Same plots with the set of parameters (B).}}
\end{figure}
The small oscillations one can observe in Fig. \ref{IMA_ASY} (B) are due to numerical errors, the exact solution is quasi constant in time, oscillating with an amplitude of approx. $10^{-7}$, as shown in the zoom of Fig. \ref{IMA_OSC} for the unknown $q_2$. The relative error in this case between the exact solution and the numerical one is of $12 \%$.
\begin{figure}[htbp]
\begin{center}
\includegraphics[width=7cm]{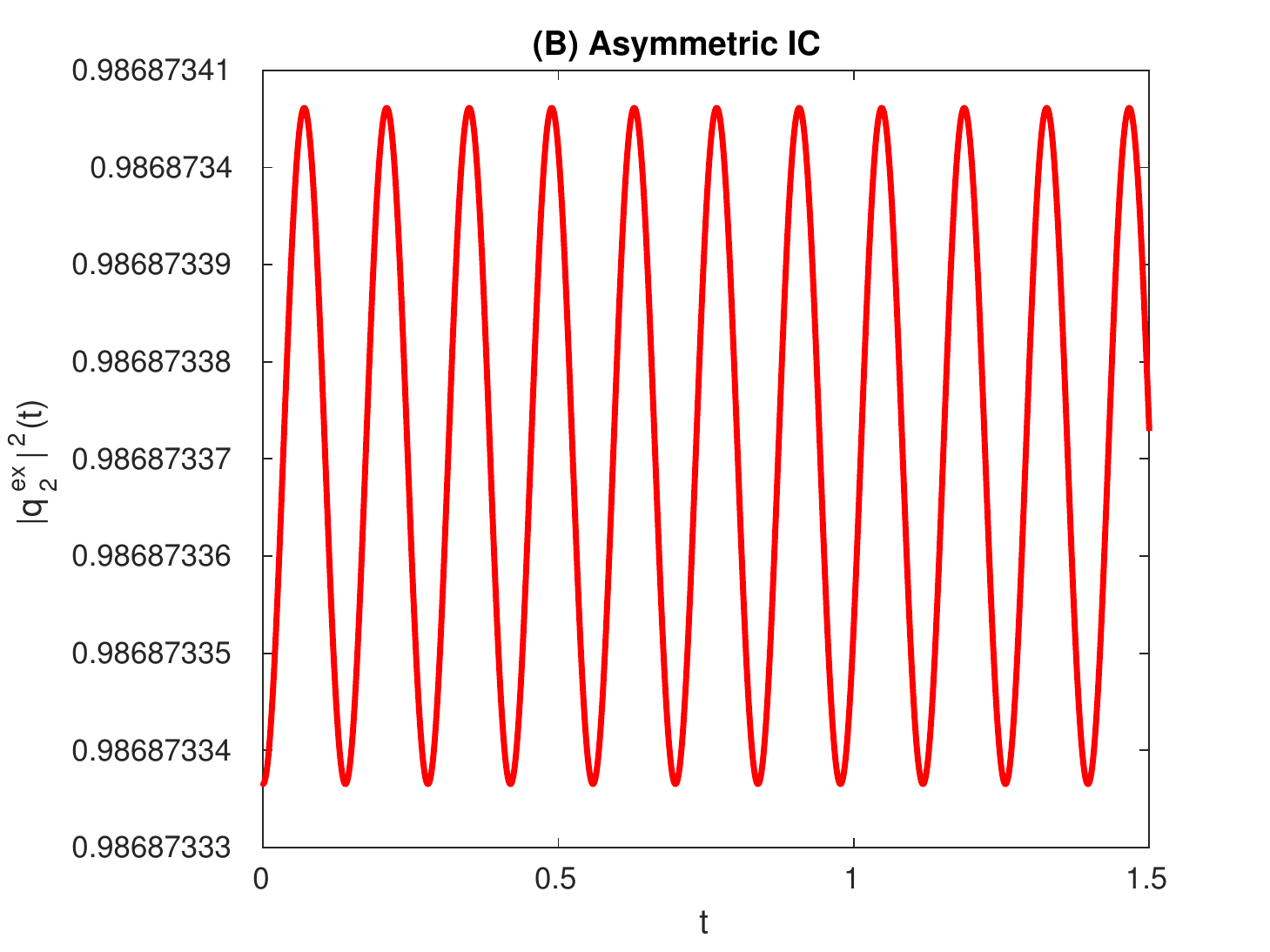}
\end{center}
\caption{\label{IMA_OSC} 
{\footnotesize The asym. lin. case. Evolution in time of the numerical solutions $|q_2|^2(t)$ with the set of parameters (B).}}
\end{figure}
\subsection{The non-linear case} \label{NON}
Let us now come to the study of the non-linear case and a detailed investigation of the destruction of the beating phenomenon. We start by choosing the same initial condition and the same parameters as in the symmetric linear case \eqref{ICB}, \eqref{PARA} and go on by raising step by step the parameter $\sigma>0$. The following Figures correspond to the following nonlinearity exponents
$$
\sigma=0.3\,; \,\, \sigma=0.6\,; \,\,\sigma=0.7\,; \,\,\sigma=0.8\,; \,\,\sigma=0.9\,; \,\,\sigma=0.98\,.
$$
What has to be mentioned here, is the choice of the parameter $\gamma$. We recall that
$$
\gamma_\pm(t)=\gamma\, |\psi(t,\pm a)|^{2 \sigma}\,.
$$
Using this formula at the initial instant $t=0$ with $\gamma_\pm(0)$ given as in the symmetric linear case, namely $\gamma_\pm(0)=-0.5$, permits after insertion of $\psi_{0}(\pm a)$ to choose $\gamma<0$ as follows
$$
\gamma:=2\,\gamma_\pm(0)/[|\psi_{0}(a)|^{2 \sigma}+|\psi_{0}(- a)|^{2 \sigma}]\,.
$$
In the following Figures we plotted the numerical solutions of the Volterra-system \eqref{VOLT2}-\eqref{SCH1D}, {\it i.e.} $|q_{1}^{num}|^2(t)$ resp. $|q_{2}^{num}|^2(t)$ (in blue resp. red) as functions of time, and for the different non-linearity exponents given above. At the same time, we plotted in the same Figures, as a reference, the exact solutions of the symmetric linear system, {\it i.e.}  $|q_{beat,1}|^2(t)$ resp. $|q_{beat,2}|^2(t)$ (in cyan resp. magenta). One observes step by step, how the non-linearity destroys the beating-effect.
\begin{figure}[htbp]
\begin{center}
\includegraphics[width=7cm]{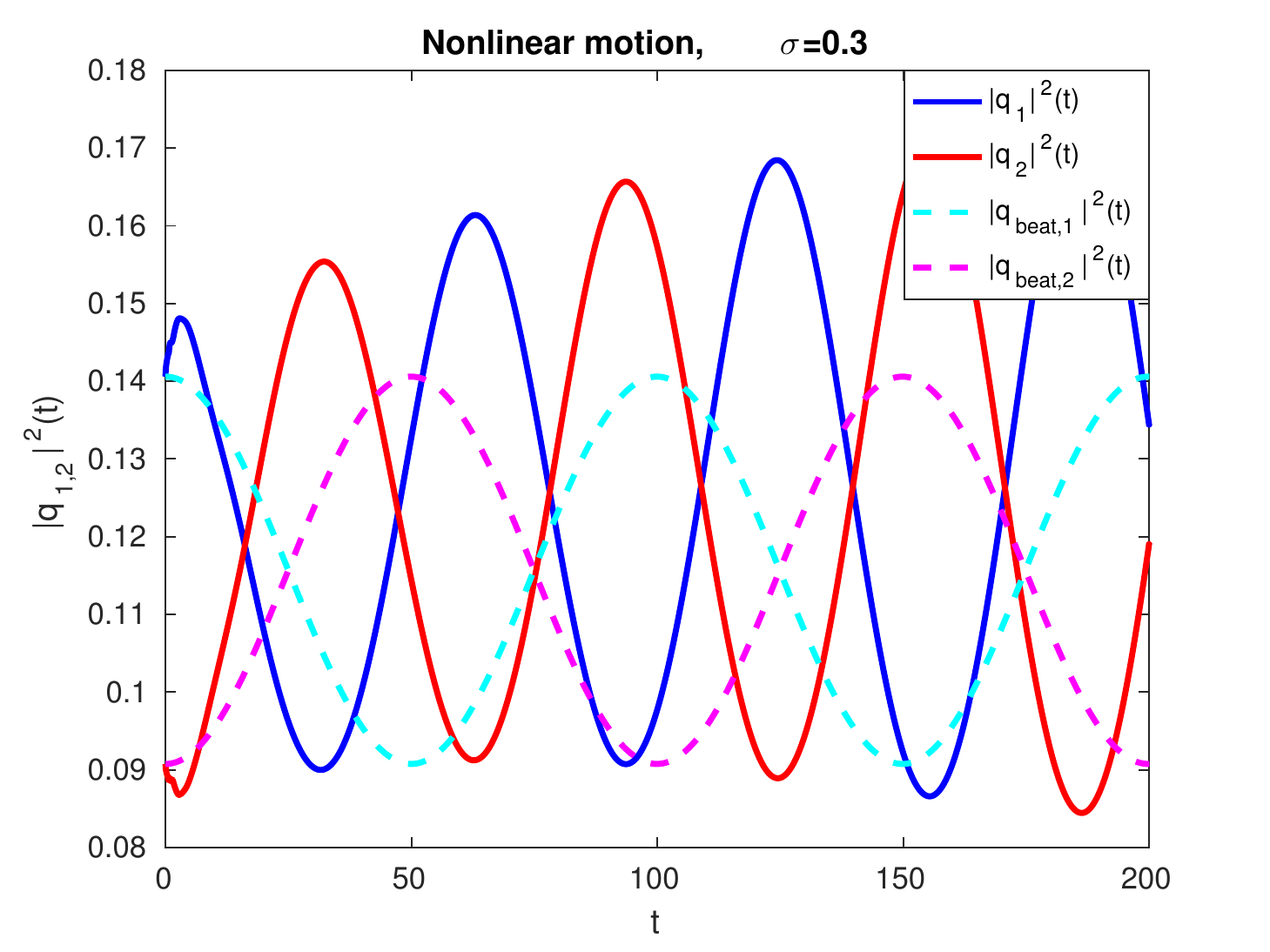}\hfill
\includegraphics[width=7cm]{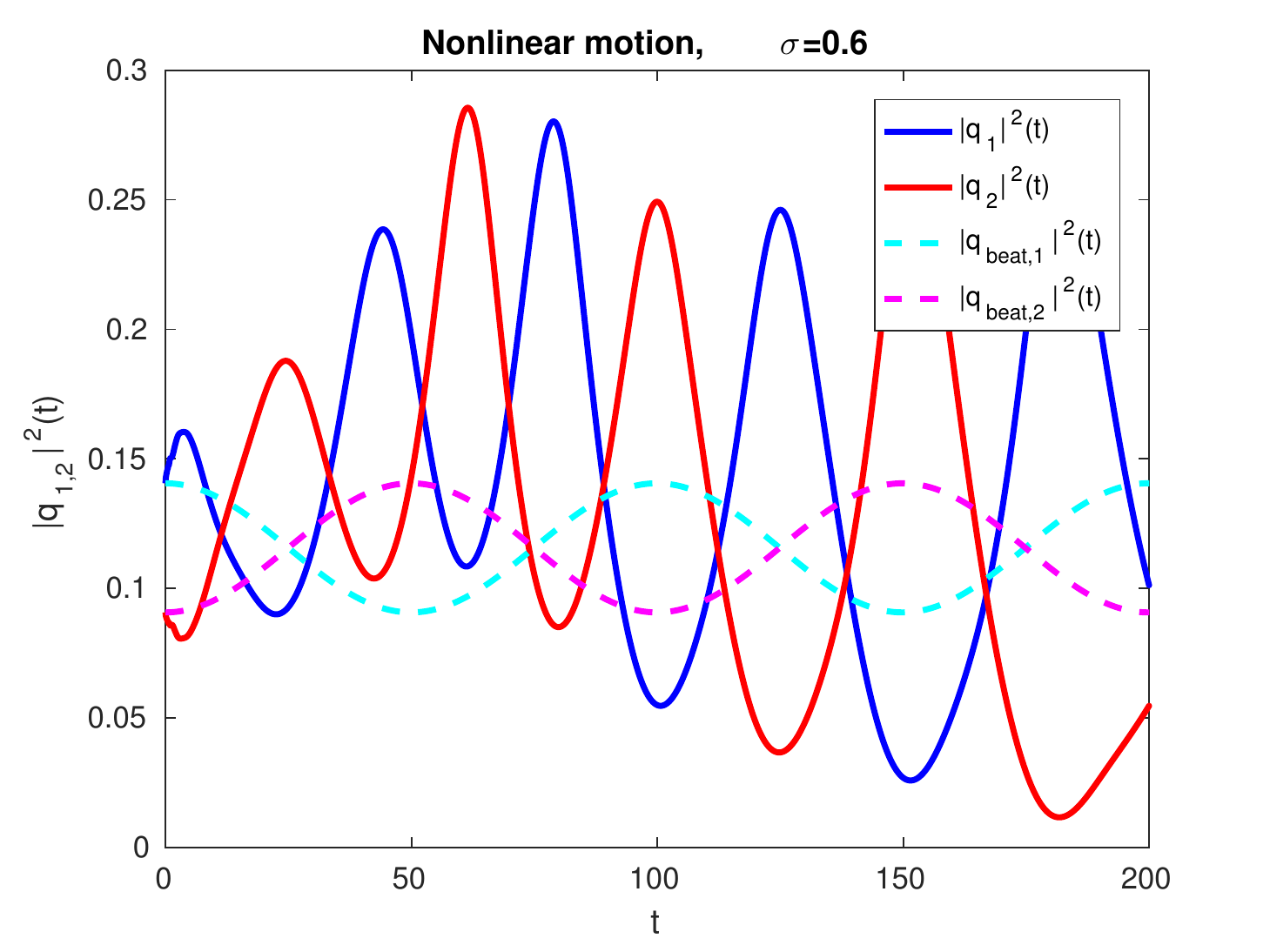}
\end{center}
\caption{\label{IMA1} 
{\footnotesize The non-linear time-evolution of the numerical solutions $|q_{1}^{num}|^2(t)$ resp. $|q_{2}^{num}|^2(t)$ (in blue/red full line) and corresponding linear beating solutions $|q_{beat,1}|^2(t)$ resp. $|q_{beat,2}|^2(t)$ (in cyan/magenta dashed line), for $\sigma=0.3$ (left) and $\sigma=0.6$ (right).}}
\end{figure}
\begin{figure}[htbp]
\begin{center}
\includegraphics[width=7cm]{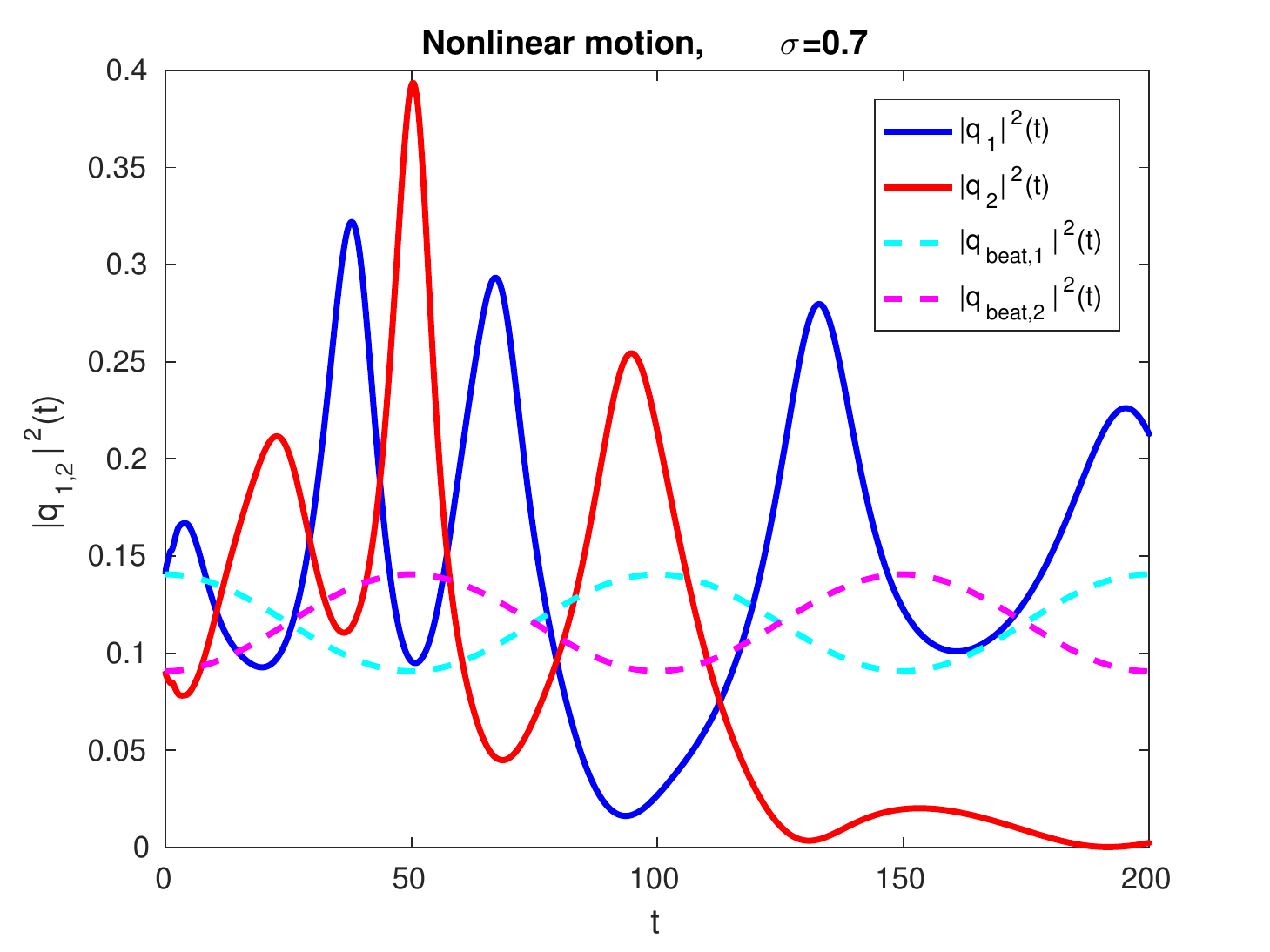}\hfill
\includegraphics[width=7cm]{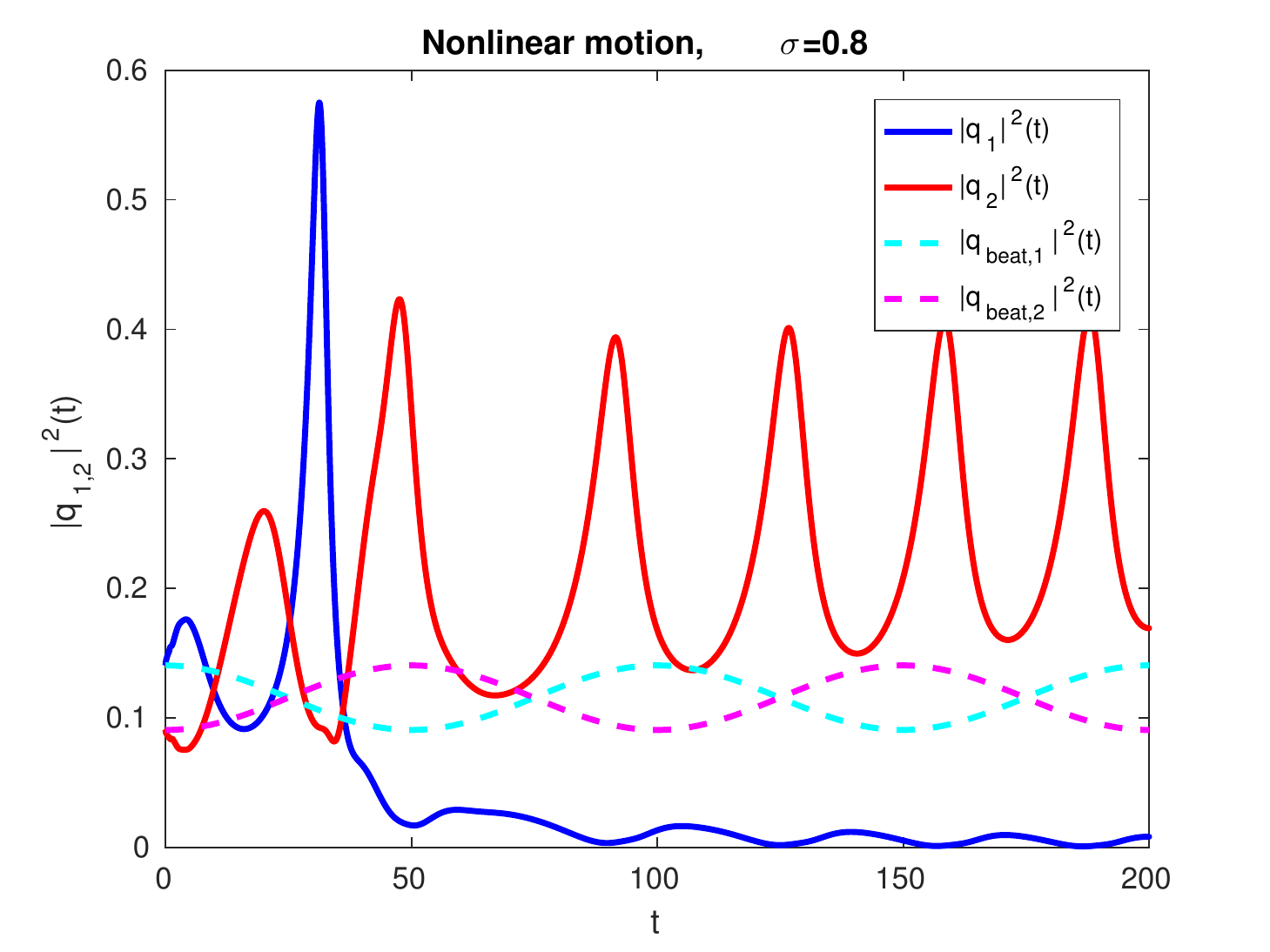}
\end{center}
\caption{\label{IMA2} 
{\footnotesize The non-linear time-evolution of the numerical solutions $|q_{1}^{num}|^2(t)$ resp. $|q_{2}^{num}|^2(t)$ (in blue/red full line) and corresponding linear beating solutions $|q_{beat,1}|^2(t)$ resp. $|q_{beat,2}|^2(t)$ (in cyan/magenta dashed line), for $\sigma=0.7$ (left) and $\sigma=0.8$ (right).}}
\end{figure}
\begin{figure}[htbp]
\begin{center}
\includegraphics[width=7cm]{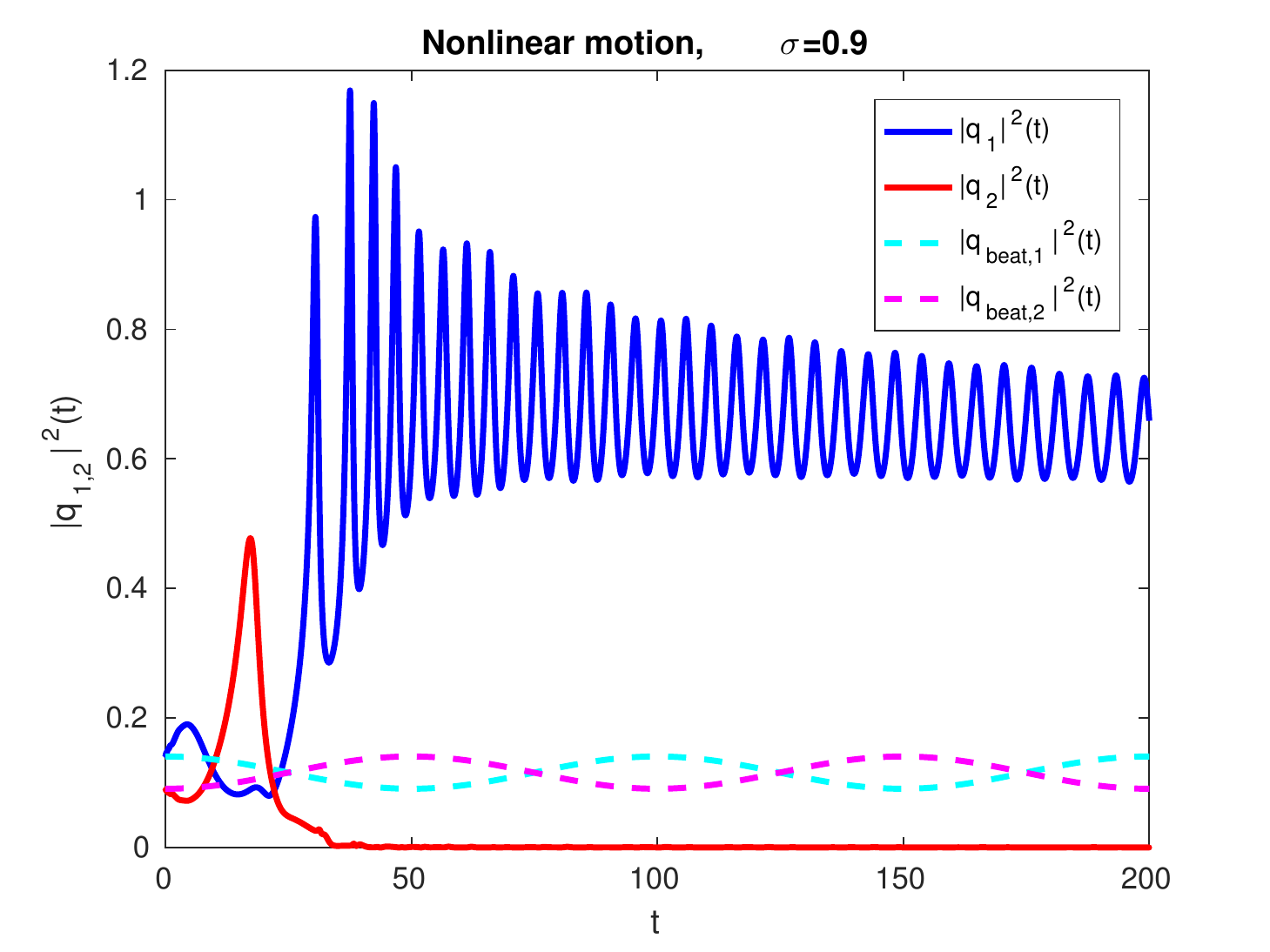}\hfill
\includegraphics[width=7cm]{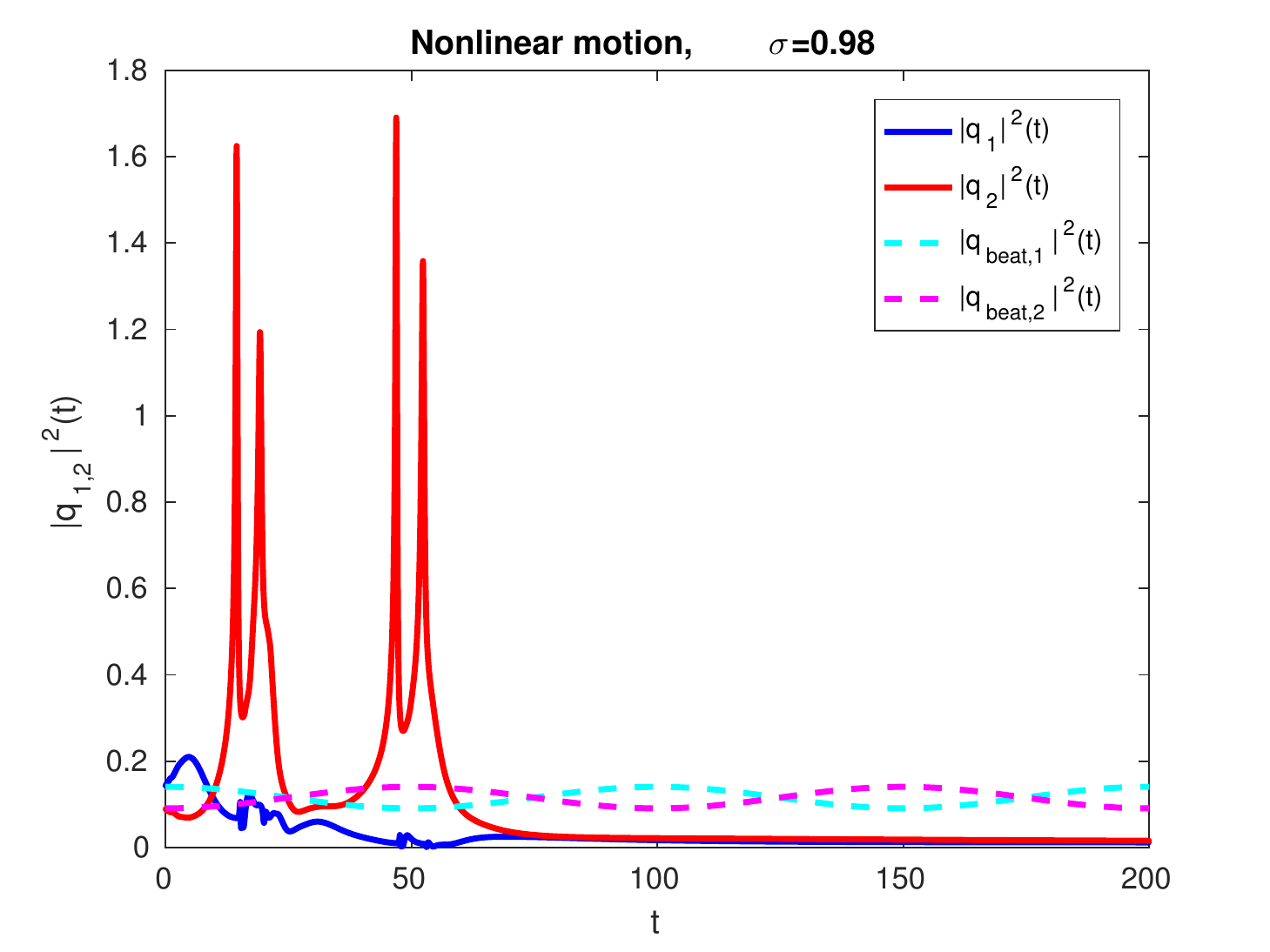}
\end{center}
\caption{\label{IMA3} 
{\footnotesize The non-linear time-evolution of the numerical solutions $|q_{1}^{num}|^2(t)$ resp. $|q_{2}^{num}|^2(t)$ (in blue/red full line) and corresponding linear beating solutions $|q_{beat,1}|^2(t)$ resp. $|q_{beat,2}|^2(t)$ (in cyan/magenta dashed line), for $\sigma=0.9$ (left) and $\sigma=0.98$ (right). Bolw-up on the right.}}
\end{figure}
\newpage
\section{Conclusion}\label{SECC}
As it was noticed by many authors in the past, the quantum beating mechanism is highly unstable under perturbations breaking the inversion symmetry of the problem. In this paper we reached similar results analyzing the suppression of the quantum beating in a zero range non-linear double well potential.

\noindent
It is worth mentioning the few features that make our model differ from the ones considered in the past.
\begin{itemize}
\item
The non-linear point interaction hamiltonians we consider here are explicitly symmetric. The asymmetry bringing to the suppression of the beating phenomenon is due to the strong dependence on  the initial conditions in the non-linear evolution.
\item
In our model no confining potential is present. As a consequence, there could be loss of mass at infinity for large times (``ionization''; see \cite{Costin:2001hp} and \cite{Correggi:2005ev}). Our results show that, in the short run, the strongly non-linear attractive double well potential cause confinement of the quantum particle and suppression of the quantum beating. The long term behavior of the solution in the non-linear case is a challenging problem that we plan to investigate in the next future.
\end{itemize}
\vspace{.2cm}
Let us conclude with few remarks on possible extensions of the present work. In this paper we chose to perform the numerical analysis of the evolution of a beating state, in presence of a non-linear perturbation of a double well potential, when the evolution equation is rephrased as  a system of two coupled weakly singular Volterra integral equation. Our aim was to test the effectiveness of this reduction in order to simplify the numerical analysis of the evolution equations. Main reason of this choice is that the very same reduction is possible in dimension two and three in spite of the fact the much more singular boundary conditions have to be satisfied at any time in those cases. As a consequence, the generalization to higher dimensions is expected to be a feasible task that we want to complete in further work.

\bigskip

\noindent {\bf Acknowledgments.} R.C.  acknowledge the support of the FIR 2013 project ``Condensed Matter in Mathematical Physics'', Ministry of University and
Research of Italian Republic  (code RBFR13WAET). C.N. would like to acknowledge
support from the CNRS-PICS project ``MANUS'' (Modelling and Numerics of Spintronics and Graphenes, 2016-2018). 

\bibliographystyle{unsrt}
\bibliography{beating.bib}





\end{document}